\documentclass[11pt]{article}
\usepackage{fullpage}
\usepackage{epigraph}
\usepackage{nicefrac}
\usepackage[margin=1in]{geometry}
\usepackage{relsize}
\usepackage{dsfont}
\usepackage{latexsym}
\usepackage[dvipsnames,table]{xcolor}
\usepackage{parskip}
\usepackage{color}
\usepackage{float}
\usepackage{bbm}
\floatstyle{boxed}
\newfloat{algorithm}{h}{lop}
\floatname{algorithm}{Algorithm}
\usepackage{tikz}
\usepackage{varwidth}
\newcommand{\Rank}{\texttt{rank}}

\usepackage{float}
\usepackage[titletoc,title]{appendix}
\usepackage[classfont=sanserif,langfont=roman,funcfont=italic]{complexity}
\usepackage{caption}
\usepackage{comment}
\usepackage{enumerate}
\usepackage{amsmath, amsthm, amssymb, amstext,  graphicx,amsopn}

\usepackage{algpseudocode}
\usepackage{geometry}

\usepackage{bbm}
\floatstyle{ruled} 
\newfloat{algorithm}{H}{lop} 
\floatname{algorithm}{Algorithm}

\usepackage{booktabs} 

\usepackage{epigraph}
\usepackage{nicefrac}
\usepackage{hyperref}
\hypersetup{colorlinks=true,linkcolor=black,citecolor=blue}

\definecolor{burgundy}{rgb}{0.5, 0.0, 0.13}
\definecolor{crimson}{rgb}{0.86, 0.08, 0.24}\usepackage{mathtools,extarrows}
\usepackage{url}
\newtheorem{theorem}{Theorem}
\newtheorem{definition}{Definition}

\newtheorem{proposition}{Proposition}

\newtheorem{lemma}{Lemma}

\newtheorem{remark}{Remark}
\allowdisplaybreaks

\definecolor{forestgreen}{rgb}{0.0, 0.27, 0.13}

\allowdisplaybreaks

\renewcommand{\vec}[1]{\mathbf{#1}}

\usepackage{thmtools}
\usepackage{thm-restate}
\usepackage[capitalise,nameinlink]{cleveref}

\crefname{figure}{Figure}{Figure}

\usepackage{mathrsfs}
\hypersetup{colorlinks=true,linkcolor=blue,citecolor=blue}
\definecolor{mygreen}{rgb}{0.0, 0.55, 0.0}
\definecolor{blue-violet}{rgb}{0.54, 0.17, 0.89}

\usepackage{tikz}
\usetikzlibrary{calc, graphs, graphs.standard, shapes, arrows, arrows.meta, positioning, decorations.pathreplacing, decorations.markings, decorations.pathmorphing, fit, matrix, patterns, shapes.misc, tikzmark}

\usepackage{amsfonts}

\renewcommand{\E}{\mathbb{E}}

\newcommand{\argmin}{\mathrm{argmin}}

\definecolor{airforceblue}{rgb}{0.36, 0.54, 0.66}
\definecolor{darkblue}{rgb}{0.0, 0.0, 0.55}

\usepackage{xcolor}
\usepackage{longtable}

\usepackage{tikz}
\usepackage[edges]{forest}
\usetikzlibrary{fit, backgrounds, shapes, calc}

\pgfdeclarelayer{background}
\pgfsetlayers{background,main}

\definecolor{bubbleyellow}{RGB}{255, 255, 230} 
\definecolor{bubbleborder}{RGB}{240, 230, 140}
\definecolor{queryred}{RGB}{220, 20, 60}      
\definecolor{graphedge}{RGB}{180, 180, 180}    

\tikzset{
    bubble/.style={
        draw=bubbleborder,
        fill=bubbleyellow,
        line width=0.8pt,
        fill opacity=0.6,
        rounded corners=8pt,
        inner sep=1.5pt 
    },
    querycircle/.style={
        draw=queryred,
        circle,
        very thick,
        inner sep=0.5pt
    }
}

\tikzset{
    querycircle/.style={
        circle,
        draw=red,
        thick,
        inner sep=2pt,
        minimum size=1.5em
    }
}

\begin{document}

\title{The Query Complexity of Local Search in Rounds on General Graphs}

\author{Simina Br\^anzei\footnote{Purdue University. E-mail: \texttt{\url{simina.branzei@gmail.com}}. Supported by US National Science
Foundation grant CCF-2238372.} \and Ioannis Panageas\footnote{University of California, Irvine. E-mail: \texttt{\url{ipanagea@ics.uci.edu}}. Supported by US National Science Foundation grant CCF-2454115.} \and Dimitris Paparas\footnote{Google Research. E-mail: \texttt{\url{dpaparas@google.com}}.}}

\date{\today}

\maketitle

\begin{abstract}
    We analyze the query complexity of finding a local minimum in $t$ rounds on general graphs. More precisely, given a graph $G = (V,E)$ and oracle access to an unknown function $f : V \to \mathbb{R}$, the goal is to find a local minimum---a vertex $v$ such that $f(v) \leq f(u)$ for all  $(u,v) \in E$---using at most $t$ rounds of interaction with the oracle. The query complexity is well understood on grids, but much less is known beyond. This abstract problem captures many optimization tasks, such as finding a local minimum of a loss function during neural network training.
    
    For each graph with $n$ vertices, we prove a deterministic upper bound of $O(t n^{1/t} (s\Delta)^{1-1/t})$, where  $s$ is the separation number and $\Delta$ is the maximum degree of the graph. We complement this result with a randomized lower bound of $\Omega(t n^{1/t}-t)$   that holds for any connected graph. We also find that parallel steepest descent with a warm start provides improved bounds for graphs with high separation number and bounded degree.

    To obtain our results, we utilized an advanced version of Gemini at various stages of our research. We discuss our experience in the Methodology section.
\end{abstract}

\section{Introduction}

Local search is a powerful heuristic for solving hard optimization problems.  Algorithms based on local search include
gradient methods, 
Lloyd’s algorithm for $k$-means clustering, the WalkSAT algorithm for Boolean satisfiability, and the Kernighan-Lin algorithm for graph partitioning. In these settings, the algorithm navigates the landscape by iteratively moving from the current configuration to a neighboring one that improves the objective. The complexity of local search is typically analyzed in two models:  white box  \cite{DBLP:journals/jcss/JohnsonPY88} and  black box \cite{aldous1983minimization}.

In the black box (query) model, there is a graph $G=(V,E)$ and an unknown function $f : V \to \mathbb{R}$ that assigns a value to each vertex.
An algorithm must query a vertex $v$ to learn $f(v)$. The goal is to return a vertex $v$ that is a local minimum: $f(v) \leq f(u)$ for all $(u,v) \in E$.

A general local search algorithm is steepest descent with a warm start~\cite{aldous1983minimization}: query $\ell$ vertices $x_1, \ldots, x_{\ell}$ chosen uniformly at random and find the vertex $x^*$ with minimal function value among these. Then run steepest descent from $x^*$,  returning the final vertex reached by the descent path. When $\ell = \sqrt{n\Delta}$, where  $n$ is the number of vertices and $\Delta$ is the maximum degree of $G$, the algorithm issues $O(\sqrt{n\Delta})$  queries in expectation and has approximately as many rounds of interaction with the oracle.

In settings such as training neural networks, each query is an expensive loss evaluation, making it crucial to parallelize computations. Motivated by these  scenarios, \cite{branzei2022query} analyzed the query complexity of local search when there are at most $t$ rounds of interaction with the oracle. An algorithm running in $t$ rounds  submits a batch of queries in each round $j$,  waits for the answers, and then submit the batch of queries for round $j+1$ \footnote{That is, the choice of queries submitted in round $j$ can only depend on the results of queries from earlier rounds.}.
At the end of the $t$-th round, the algorithm stops and outputs a solution.

The analysis in \cite{branzei2022query} focused on the $d$-dimensional grid,  leaving open the question of understanding complex, non-Euclidean geometries, which are central to many modern optimization tasks.
In this paper, we address the question for general graphs. Next we give several examples (see also \cite{branzei2022query}) of applications captured by this abstract model.

\subsection{Examples}

\paragraph{Linear Regression with Non-Convex Regularization.}
Suppose we are given a dataset of $m$ labeled examples  $\{(\vec{a}_i,b_i)\}_{i=1}^{m} \subseteq \mathbb{R}^{d} \times \mathbb{R}$. The goal is to find a vector of coefficients $\vec{x} \in \mathbb{R}^d$ that minimizes the loss function:
$
L(\vec{x}) = \frac{1}{m}\sum_{i=1}^{m} (\langle\vec{a}_i,\vec{x}\rangle - b_i)^2 + \sum_{j=1}^d P_{\lambda}(x_j).
$
Here, the first term is the mean squared error, and $P_\lambda$ is a non-convex penalty such as the Minimax Concave Penalty~\cite{zhang2010mcp}. 

Although the ideal coefficients exist in the continuous domain $\mathbb{R}^d$, numerical solvers inherently operate via discrete updates within a bounded region, such as $[-B, B]^d$. We formalize this search space as a grid graph $[n]^d$, where each vertex $\vec{x}$ represents a vector of candidate regression coefficients. In this graph, an edge connects two vertices if they differ by a fixed discretization step in exactly one coordinate. A query at vertex $\vec{x}$ reveals the loss $f(\vec{x}) = L(\vec{x})$.


\paragraph{Hyperparameter Optimization.} In hyperparameter optimization for deep neural networks, the domain is the high-dimensional space of hyperparameters and the function $f$ is the validation error of the network. This typically induces a non-convex landscape. 

Unlike simple regression, a single ``query'' corresponds to training a deep neural network to convergence, which is an extremely expensive operation (taking hours or days).
Consequently, sequential search is often infeasible. Instead, practitioners use parallel computing resources to train multiple configurations simultaneously—constituting a single round. Thus minimizing the number of rounds is critical for reducing the total time to solution.

\paragraph{Robust Matrix Estimation (General Graphs).}
In robust matrix estimation,  noisy data is collected into a matrix $M \in \mathbb{R}^{n \times n}$ where we only observe a subset of entries---for example, $M_{ij}$ is the rating user $i$ gives to movie $j$. Let $\Omega \subseteq [n] \times [n]$ denote the set of indices we observed.
Since the observations in $M$ may contain errors, we do not want to match them exactly. Instead, we assume the true underlying preferences form a simple (low-rank) structure. The goal is to find a matrix $X$ that approximates the observations in $M$ while adhering to this structure:
\[
    \min_{X} \| P_\Omega(X - M) \|_1 \quad \text{subject to} \quad \text{rank}(X) \le r,
\]
where $P_\Omega(\cdot)$ is the projection operator that preserves entries in $\Omega$ and zeros out the rest\footnote{That is, $[P_\Omega(A)]_{ij} = A_{ij}$ if $(i,j) \in \Omega$ and $0$ otherwise.}.

Because we cannot search this continuous space with infinite precision, we consider a discrete set of candidate solutions forming a graph $G=(V,E)$:

\begin{description}
\item \emph{The Nodes (Rank-$r$ Matrices):}
    Each node $v \in V$ represents a  candidate matrix $X$. Since the rank is constrained, we parameterize the solution as $X = U W^T$, where $U, W \in \mathbb{R}^{n \times r}$.

    \item \emph{The Edges (Atomic Perturbations):}
    Consider an arbitrary node defined by $X = UW^T \in V$. A neighbor $X'$ is generated by perturbing a single entry of a factor matrix (say $U$) by a scalar step size $\delta$. 
    That is, let $U'$ be a matrix that is identical to $U$ everywhere, except for the entry at index $(i,k)$ defined as:
    $ 
        U'_{ik} = U_{ik} + \delta \,.
    $ 
    The neighbor node is the product $X' = U' W^T$.
    
    Since $U' = U + \delta \mathbf{e}_i \mathbf{h}_k^T$ (where $\mathbf{e}_i \in \mathbb{R}^n$ and $\mathbf{h}_k \in \mathbb{R}^r$ are standard basis vectors), the neighbor relates to $X$ via a rank-1 update:
    $
        X' = (U + \delta \mathbf{e}_i \mathbf{h}_k^T) W^T = X + \delta (\mathbf{e}_i \mathbf{w}_k^T).
    $
    Here $\mathbf{w}_k$ is the $k$-th column of $W$ and the term $\delta (\mathbf{e}_i \mathbf{w}_k^T)$ is a matrix where the $i$-th row is non-zero and all other rows are zero. Consequently, changing a single entry $U_{ik}$ modifies the {entire} $i$-th row of $X$.
\end{description}

This graph is topologically distinct from a grid since it contains triangles\footnote{To see that the matrix graph contains triangles, let $X = UW^T$ be a candidate solution. Let $Y = U'W^T$ be the neighbor where $U'$ is identical to $U$ except for the entry $U'_{1,1} = U_{1,1} + \delta$. Similarly, let $Z = U''W^T$ be the neighbor where $U''$ is identical to $U$ except for the entry $U''_{1,1} = U_{1,1} + 2\delta$. All three solutions share the same factor matrix $W$. Since the factor matrices $U'$ and $U''$ differ by exactly $\delta$ in entry $(1,1)$ (and are identical elsewhere), the nodes $Y$ and $Z$ are also connected by an edge. Thus $X, Y,$ and $Z$ form a triangle.}.

\medskip 

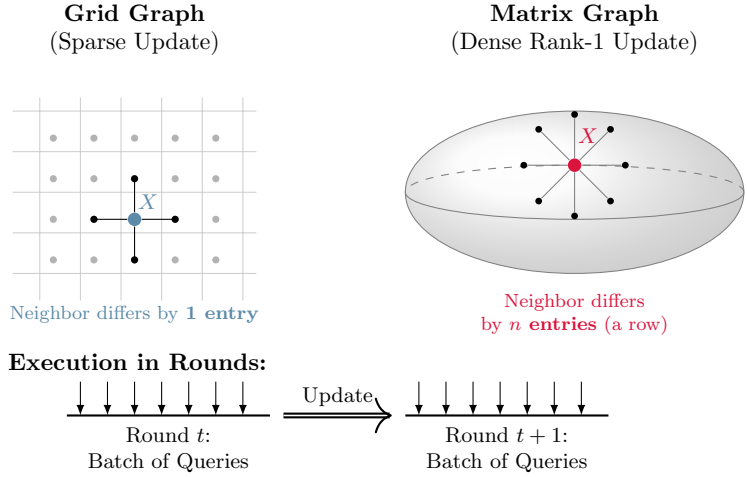
\begin{figure}[H]
    \centering
    \begin{tikzpicture}[scale=0.9, every node/.style={transform shape}]
       
        \node[align=center, font=\small] at (2, 4.2) {\textbf{Grid Graph}\\(Sparse Update)};
        \draw[step=0.6cm, gray!40, very thin] (0.2,0.2) grid (3.8,3.2);
       
        \foreach \x in {0.8, 1.4, ..., 3.2} \foreach \y in {0.8, 1.4, ..., 2.6}
            \fill[gray!60] (\x,\y) circle (1.5pt);
            
        \foreach \nx/\ny in {2.0/2.0, 2.6/1.4, 2.0/0.8, 1.4/1.4}
            \fill[black] (\nx,\ny) circle (1.5pt);
       
        \node[circle, fill=airforceblue, inner sep=2pt] (start_grid) at (2.0, 1.4) {};
       
        \draw[black] (start_grid) -- (2.0, 2.0); 
        \draw[black] (start_grid) -- (2.6, 1.4); 
        \draw[black] (start_grid) -- (2.0, 0.8); 
        \draw[black] (start_grid) -- (1.4, 1.4); 
       
        \node[below, airforceblue, font=\footnotesize] at (2.2, 1.9) {$X$};
        \node[align=right, font=\scriptsize, airforceblue] at (2, 0) {Neighbor differs by \textbf{1 entry}};

        \node[align=center, font=\small] at (8.5, 4.2) {\textbf{Matrix Graph}\\(Dense Rank-1 Update)};
       
        \shade[ball color=gray!10, opacity=0.3] (8.5, 1.8) ellipse (2.5cm and 1.2cm);
        \draw[gray] (8.5, 1.8) ellipse (2.5cm and 1.2cm);
        \draw[gray] (6, 1.8) arc (180:360:2.5 and 0.4);
        \draw[gray, dashed] (11, 1.8) arc (0:180:2.5 and 0.4);

        \node[circle, fill=crimson, inner sep=2pt] (center) at (8.5, 2.2) {};
        \node[above, crimson, font=\footnotesize] at (8.7, 2.4) {$X$};
       
        \foreach \angle in {0, 45, ..., 315} {
            \node[circle, fill=black, inner sep=1pt] (n\angle) at ($(center)+(\angle:0.75cm)$) {};
            \draw[thin, gray] (center) -- (n\angle);
        }
           
        \node[align=center, font=\scriptsize, crimson] at (8.5, 0) {Neighbor differs\\by \textbf{$n$ entries} (a row)};

        \node[anchor=west, font=\small] at (0, -0.7) {\textbf{Execution in Rounds:}};
       
        \draw[thick] (1, -1.5) -- (4, -1.5);
        \foreach \x in {1.2, 1.6, 2.0, 2.4, 2.8, 3.2, 3.6} \draw[->, >=latex] (\x, -1.0) -- (\x, -1.5);
        \node[align=center, font=\footnotesize] at (2.5, -2.0) {Round $t$: \\ Batch of Queries};

        \draw[->, thick, double] (4.2, -1.5) -- (5.8, -1.5);
        \node[above, font=\footnotesize] at (5, -1.5) {Update};

        \draw[thick] (6, -1.5) -- (9, -1.5);
        \foreach \x in {6.2, 6.6, 7.0, 7.4, 7.8, 8.2, 8.6} \draw[->, >=latex] (\x, -1.0) -- (\x, -1.5);
        \node[align=center, font=\footnotesize] at (7.5, -2.0) {Round $t+1$: \\ Batch of Queries};

    \end{tikzpicture}
    \caption{\emph{Left:} In the Grid graph, a neighbor of $X$ differs from $X$ in a single coordinate. Right: In the Matrix graph, a neighbor of $X$ differs from $X$ in multiple coordinates (specifically, an entire row). This dense connectivity results in a high-expansion graph, distinct from a grid. \emph{Bottom:} To mitigate high probe latency, the algorithm queries batches in parallel.}
    \label{fig:motivation}
\end{figure}

\subsection{Graph Features and Complexity of Local Search}

To build intuition on the relation between the geometry of the graph and query complexity, consider the impact of the \emph{Minimum Vertex Cover}. 
If graph $G=(V,E)$ has a minimum vertex cover\footnote{A set $C$ of vertices is a vertex cover if every edge in $E$  is incident to at least one vertex in $C$.} $C$, then we can find a local minimum in just two rounds and $|C| + \Delta$ queries:

%
\begin{itemize}
    \item \emph{Round 1:} Query all vertices in the minimum vertex cover $C$. Let $Z$ be the vertex in $C$ with the minimum observed value.
    \item \emph{Round 2:}  Query all the neighbors of $Z$ in $V\setminus C$. If $Z$ is a local minimum, output it. Otherwise, one of the neighbors of $Z$ must be a local minimum, return it.
\end{itemize}
If $Z$ is smaller than all its neighbors, it is a local minimum. Otherwise, let $w$ be the neighbor of $Z$ with the smallest function value $f(w) < f(Z)$. We claim $w$ is a local minimum. If $w$ had a neighbor $u$ with $f(u) < f(w)$, then $u$ cannot be in $C$ (otherwise $Z$ would not be the minimum in $C$). Thus $u \in V \setminus C$. However, this implies the edge $(u,w)$ connects two vertices in $V \setminus C$, contradicting the definition of a vertex cover (which must cover every edge). Thus, $w$ is a local minimum. The query complexity is $c+\Delta$, where $\Delta$ is the maximum degree of the graph.

This example illustrates how geometric structure can enable pruning the search space. In the remainder of the paper, we generalize this intuition to $t$-round algorithms, bounding the query complexity via the \emph{separation number} of the graph and giving randomized lower bounds. 

\subsection{Model} \label{sec:model}

Let $G = (V,E)$ be a connected undirected  graph and $f : V \to \mathbb{R}$ a function.  
A vertex $v \in V$  is a local minimum if $f(v) \leq f(u)$ for all $\{u,v\} \in E$.
We write $V = [n] = \{1, \ldots, n\}$.
Given as input a graph $G$ and oracle access to an unknown function $f$, the local search problem is to find a local minimum of $f$ on $G$ using as few queries as possible. {Each query is of the form: ``Given a vertex $v$, what is $f(v)$?''}.

 Let $t$ be an upper bound on the number of rounds of interaction with the oracle. An algorithm running in $t$ rounds  submits in each round $j$ a set of queries,  waits for the answers, and then submits the set of queries for round $j+1$ \footnote{That is, the choice of queries submitted in round $j$ can only depend on the results of queries from earlier rounds.}.
At the end of the $t$-th round, the algorithm  stops and outputs a solution.

\paragraph{Query Complexity.} The \emph{deterministic query complexity} is the total number of queries necessary and sufficient for an optimal deterministic algorithm to find a solution on a worst case input function.
The \emph{randomized query complexity} is the minimum worst-case number of queries required by a randomized algorithm to compute the function with probability at least $9/10$ for every input.

\paragraph{Graph Features.}

Let $\Delta$ denote the maximum degree of the graph $G$. 
For each $u,v \in V$, let $dist(u,v)$ be the length of the shortest path from $u$ to $v$.

Let ${1}/{2} \le \alpha < 1$ be a real number, $s \in \mathbb{N}$, and $G=(V,E)$ a graph.
A subset $S \subseteq V$ is an \emph{$(s, \alpha)$-separator} of $G$, if there exist disjoint subsets $A, B \subseteq V$ such that the next properties hold:
(i) $V = A \cup B \cup S$; 
(ii) $|S| \le s$ and $|A|, |B| \le \alpha|V|$; and 
(iii) $E(A,B) = \emptyset$.

The \emph{separation number $s(G)$} of $G$ is the smallest $s$ such that all subgraphs $G'$ of $G$ have an $(s, 2/3)$-separator.
The separation number is within a constant factor of the treewidth. For additional discussion, see   chapter 7 of \cite{Parameterized_algos_book} and \cite{graph_notions_survey}.

\subsection{Our Results} \label{sec:our_results}

First, we consider a divide-and-conquer approach based on recursively finding separators of the initial graph $G$ and has good performance when the separation number is sublinear in $n$.

\begin{theorem} \label{thm:k_rounds}
Let $G = (V, E)$ be a connected undirected graph with $n$ vertices.
The  deterministic query complexity of finding a local minimum on $G$ in $t \ge 2$ rounds is at most $\min(4t   n^{\frac{1}{t}}  (s\Delta)^{1-\frac{1}{t}},n  )$,
where  $\Delta$ is the maximum degree and $s$ is the separation number of $G$.
\end{theorem}

For graphs with large separation number, a randomized approach that uses parallel steepest descent with a warm start can be better. This is based on the classical steepest descent with a warm start method~\cite{aldous1983minimization}.

\begin{proposition} \label{thm:t_rounds_const_deg_generalized}
    Let $G=(V,E)$ be a graph with $n$ vertices and maximum degree $\Delta$. The randomized query complexity of finding a local minimum in $t \ge 2$ rounds is $O(\sqrt{n} + t)$ when $\Delta \leq 2$ and $O\bigl(\frac{n}{t \cdot \log_\Delta n} + t \Delta^2 \sqrt{n}\bigr)$ when $\Delta \geq 3$.
\end{proposition}

The upper bound is  $O(n/\log{n})$ in two rounds for graphs with bounded maximum degree. The algorithm analyzed in Proposition~\ref{thm:t_rounds_const_deg_generalized} resembles the one for general graphs in \cite{zhang2009tight} (section 5)   and the fractal-like steepest descent from \cite{branzei2022query} for the $d$-dimensional grid.

We also obtain the following lower bound.

\begin{theorem} \label{thm:lb_t_rounds}
    Let $G = (V,E)$ be a connected undirected graph with $n$ vertices. The randomized query complexity of finding a local minimum on $G$ in $t \in \mathbb{N}^*$ rounds is $\Omega(t n^{1/t} - t)$.
\end{theorem}

The constant hidden in $\Omega$ is independent of $n$ and $t$. The proof of Theorem~\ref{thm:lb_t_rounds} shows a lower bound of $\Omega(c \cdot t n^{\frac{1}{t}} - t)$ for each success probability $c \in (1/n, 1]$.

The proof of Theorem~\ref{thm:lb_t_rounds} fixes a spanning tree of the graph, rooted at some vertex $r$. Then it defines a hard distribution using a family of functions based on  a staircase construction. A vertex $Z$ is chosen uniformly at random from $V$ to represent the target local minimum. Given a choice of $Z$, for each vertex $v \in V$, let $dist_T(r,v)$ represent the distance in $T$ between $r$ and $v$. If $v$ is on the path from $Z$ to $r$, then the function is defined as $f(v) = -dist_T(r,v)$. Otherwise, $f(v) = dist_T(r,v)$.

Analyzing this distribution provides a lower bound of $\Omega(t n^{1/t} - t)$. By ensuring that every round reduces the search space to a sub-tree---leaving a smaller instance of the original problem---we can employ a clean inductive argument.

\section{Related Work}

\paragraph{The Boolean Hypercube and Grids.}
The query complexity of local search was first studied theoretically by Aldous~\cite{aldous1983minimization}, who analyzed the Boolean hypercube $\{0,1\}^n$. He analyzed the ``steepest descent with a warm start'' heuristic, showing it requires $O(\sqrt{n} \cdot 2^{n/2})$ queries, and established a nearly matching lower bound of $\Omega(2^{n/2-o(n)})$ using a random walk construction based on hitting times. This lower bound was subsequently refined by Aaronson~\cite{Aaronson06} via the relational adversary method, and later tightened by Zhang~\cite{zhang2009tight} to $\Theta(\sqrt{n} \cdot 2^{n/2})$. For deterministic algorithms, Llewellyn, Tovey, and Trick~\cite{llewellyn1989local} provided a divide-and-conquer strategy that yields an upper bound of $O(2^n \log n / \sqrt{n})$ on the hypercube.

For the $d$-dimensional grid $[n]^d$, randomized lower bounds were established by Aaronson~\cite{Aaronson06} and Zhang~\cite{zhang2009tight}, with the latter proving a tight $\Omega(n^{d/2})$ bound for constant $d \geq 4$. Sun and Yao~\cite{sun2009quantum} resolved remaining gaps for low dimensions ($d=2,3$) and quantum settings.

\paragraph{General Graphs.}
For arbitrary graphs, complexity is often characterized by structural invariants.  Santha and Szegedy~\cite{santha2004quantum} utilized the graph's \emph{separation number} $s$ to derive a quantum lower bound of $\Omega(\sqrt[8]{s/\Delta}/\log n)$ and a deterministic upper bound of $O((s+\Delta)\log n)$.
Other works have linked query complexity to topological features such as genus~\cite{Verhoeven06} and diameter in Cayley and vertex-transitive graphs~\cite{dinh2010quantum}. More recently, lower bounds have been derived from spectral properties, including graph congestion~\cite{BCR23} and mixing time~\cite{BR_2024spectrallowerboundslocal} of the fastest mixing Markov chain for the given graph.

\paragraph{Local Search in Rounds and Distributed Settings.}
The specific setting of parallel query rounds (adaptivity) was analyzed by Br\^anzei and Li~\cite{branzei2022query} for grid graphs, providing bounds for both constant and polynomial number of rounds. In the distributed setting, Babichenko, Dobzinski, and Nisan~\cite{babichenko2019communication} studied the communication complexity of finding local minima, which captures settings in the cloud, where data is held by different parties.

\paragraph{Stationary Points and Adaptive Complexity.}
Discrete local search is closely related to finding approximate stationary points in continuous optimization (i.e., points $x$ where $\|\nabla f(x)\| \leq \epsilon$). While gradient descent is inherently sequential, recent work has investigated the limits of parallelization in this domain. For examples of works on algorithms and complexity of computing approximate stationary points, see, e.g., \cite{vavasis1993black,pmlr-v119-zhang20p,  stationary_I,stationary_II,bubeck2020trap,pmlr-v119-drori20a}).

Zhou et al.~\cite{Zhou2025adaptive} explicitly analyzed the ``adaptive complexity'' of finding stationary points. They demonstrated a dichotomy based on dimension: in high-dimensional settings ($d \approx \text{poly}(1/\epsilon)$), parallelization offers no asymptotic benefit over sequential methods. Conversely, for constant dimensions, they developed an algorithm bridging grid search and gradient flow trapping~\cite{bubeck2020trap} that achieves near-optimal query complexity in constant rounds. Their lower bound analysis for constant dimensions relies on a reduction to the discrete local search problem on grid graphs, highlighting the tight connection between these continuous and discrete models.

\section{Algorithms} \label{sec:algorithms}

In this section we study the query complexity of local search as a function of the separation number of the graph. Lower bounds on the query complexity of local search via separation number were first provided by \cite{santha2004quantum} for fully adaptive algorithms.

\subsection{Deterministic Algorithm}

We use  a recursive application of the separation property, summarized by the next folklore lemma.

\begin{lemma}[Shattering Lemma] \label{lem:shattering}
Let $G=(V,E)$ be a graph with $n$ vertices and separation number $s$. For any parameter $K \in [1, n]$, there exists a subset of vertices $S \subseteq V$  such that every connected component of the induced subgraph $G[V \setminus S]$ has size at most $K$, and  $|S| < 3sn/K$.
\end{lemma}

As a warm-up, we first sketch an algorithm for two rounds and then extend it to any number of rounds. Let $K \in \{1, \ldots, n\}$ to be set later.
 Let $S$ be the separator guaranteed by Lemma~\ref{lem:shattering} with parameter $K$ and $\mathcal{C} = \{C_1, \ldots, C_m\}$  the set of connected components of $G[V \setminus S]$.

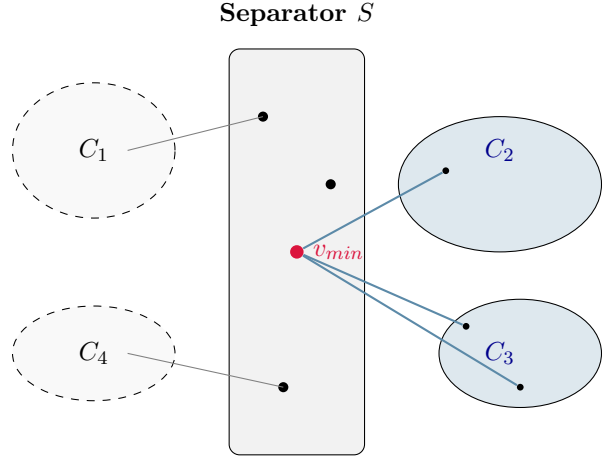
\begin{figure}[H]
    \centering
    \begin{tikzpicture}[scale=0.9, every node/.style={transform shape}]
        \draw[fill=gray!10, rounded corners] (3, 0) rectangle (5, 6);
        \node[align=center, font=\bfseries] at (4, 6.3) {Separator $S$\\};
       
        \node[circle, fill=black, inner sep=1.5pt] at (3.5, 5) {};
        \node[circle, fill=black, inner sep=1.5pt] at (4.5, 4) {};
        \node[circle, fill=black, inner sep=1.5pt] at (3.8, 1) {};
       
        \node[circle, fill=crimson, inner sep=2pt, label={[crimson]right:$v_{min}$}] (vmin) at (4, 3) {};
       
        \draw[dashed, fill=gray!5] (1, 4.5) ellipse (1.2cm and 1cm);
        \node at (1, 4.5) {$C_1$};
        \draw[thin, gray] (3.5, 5) -- (1.5, 4.5); 
       
        \draw[fill=airforceblue!20] (7, 4) ellipse (1.5cm and 1cm);
        \node[darkblue] at (7, 4.5) {$C_2$};
       
        \node[circle, fill=black, inner sep=1pt] (n1) at (6.2, 4.2) {};
        \draw[thick, airforceblue] (vmin) -- (n1);
       
        \draw[fill=airforceblue!20] (7.3, 1.5) ellipse (1.2cm and 0.8cm);
        \node[darkblue] at (7, 1.5) {$C_3$};
       
        \node[circle, fill=black, inner sep=1pt] (n2) at (6.5, 1.9) {};
        \draw[thick, airforceblue] (vmin) -- (n2);
        \node[circle, fill=black, inner sep=1pt] (n3) at (7.3, 1) {};
        \draw[thick, airforceblue] (vmin) -- (n3);
       
        \draw[dashed, fill=gray!5] (1, 1.5) ellipse (1.2cm and 0.7cm);
        \node at (1, 1.5) {$C_4$};
        \draw[thin, gray] (3.8, 1) -- (1.5, 1.5);

    \end{tikzpicture}
    \caption{Visual representation of the two-round algorithm. In Round 1, the separator $S$ is queried to find $v_{min}$. In Round 2, the algorithm only queries the components ($C_2, C_3$) containing neighbors of $v_{min}$.  Components ($C_1, C_4$) are not connected to $v_{min}$, so they are ignored.}
    \label{fig:2_round_algo}
\end{figure}
 
 \refstepcounter{algorithm} \label{alg:algorithm_two_rounds}
\noindent\textbf{Algorithm \thealgorithm.}

\begin{itemize}
\item \textbf{Round 1:} Query all vertices in $S$ and identify the global minimum among them, denoted  $v_{min} = \argmin_{v \in S} f(v)$.

\item \textbf{Round 2:} Identify the components $C_i$ that contain vertices adjacent to $v_{min}$ and query all the vertices in these components.

 \emph{Output:}
     If $f(v_{min}) \le f(u)$ for all $u \in N(v_{min})$, output $v_{min}$. Else,
    let $C^*$ the component in $\mathcal{C}$ that contains the smallest neighbor of $v_{min}$.
    Output $v^* = \argmin_{v \in C^*} f(v)$.
\end{itemize}
As we show later, the algorithm always outputs a local minimum. Setting $K \approx \sqrt{3sn/\Delta}$ ensures that the number of queries in round $1$ is $|S| \leq \sqrt{3sn\Delta}$ and the number of queries in round $2$ is at most $\Delta K \approx \sqrt{3sn\Delta}$, which leads to  $O(\sqrt{sn\Delta})$ queries in total.

The algorithm for $t$ rounds is a generalization of the two-round strategy. It  recursively invokes the shattering lemma and selects separator sizes to equalize the work done across different rounds.

\medskip

\noindent \textbf{Theorem~\ref{thm:k_rounds}.}
\emph{Let $G = (V, E)$ be a connected undirected graph with $n$ vertices.
The  deterministic query complexity of finding a local minimum on $G$ in $t \ge 2$ rounds is at most $\min(4t   n^{\frac{1}{t}}  (s\Delta)^{1-\frac{1}{t}},n  )$,
where  $\Delta$ is the maximum degree and $s$ is the separation number of $G$.}
\begin{proof}
Let $Q_t(G)$ denote the randomized query complexity of finding a local minimum on $G$.
We know $Q_t(G) \le n$. If $t=1$, then the query complexity is trivially at most $n$ and the formula holds. Assume $t \ge 2$.
We design a deterministic $t$-round algorithm using a $(t-1)$-level hierarchical decomposition based on the Shattering Lemma (Lemma~\ref{lem:shattering}). We define parameters $K_1, \ldots, K_{t-1} \in \mathbb{N}$ such that $n \ge K_1 \ge \ldots \ge K_{t-1} \ge 1$.
If multiple vertices have the same function value, we break ties in lexicographic order of vertex indices.

\newpage 

 \refstepcounter{algorithm} \label{alg:t_rounds}
\noindent\textbf{Algorithm \thealgorithm: Separator-based algorithm in $t$ rounds.}

\paragraph{\textit{Preprocessing Step (Hierarchical Decomposition).}}
\begin{enumerate}
    \item \textbf{Level 1:} Apply Lemma~\ref{lem:shattering} to $G$ with parameter $K_1$. This yields the primary separator $S_1$ and the set of connected components $\mathcal{C}_1$ of $G[V \setminus S_1]$.
    \item \textbf{Level $i$ ($2 \le i \le t-1$):} For each component $C_{i-1} \in \mathcal{C}_{i-1}$, apply Lemma~\ref{lem:shattering} to $G[C_{i-1}]$ with parameter $K_i$. This yields a separator $\sigma_i(C_{i-1}) \subseteq C_{i-1}$ and sub-components $\mathcal{C}_i(C_{i-1})$.
    Let $\mathcal{C}_i$ be the collection of all level-$i$ components \footnote{\emph{Remark:} If  $|C_{i-1}| \le K_i$ for some index $i$, then Lemma~\ref{lem:shattering} yields $\sigma_i(C_{i-1}) = \emptyset$, and the component remains intact.}.
\end{enumerate}

We define the level of a vertex $v$: Level$(v)=i$ if $v$ is in a level-$i$ separator ($1\le i \le t-1$), and Level$(v)=t$ if $v$ is in a final component $C \in \mathcal{C}_{k-1}$.

\paragraph{\textit{Execution.}}
Let $\mathcal{Q} := \emptyset$ denote the set of vertices queried so far by the algorithm. We maintain the running minimum $v_i$ among them.
\begin{itemize}
\item \textbf{Round 1:} Query $Q_1 = S_1$. Set $\mathcal{Q} \leftarrow Q_1$ and  $v_1 = \argmin_{v \in \mathcal{Q}} f(v)$.

\item \textbf{Round $i$ ($2 \le i \le t-1$):}
    \begin{enumerate}
                \item Let $\mathcal{A}_{i-1} \subseteq \mathcal{C}_{i-1}$ denote the set of level-$(i-1)$ components containing neighbors of $v_{i-1}$.

        \item Query the level-$i$ separators within these components: $Q_i = \bigcup_{C \in \mathcal{A}_{i-1}} \sigma_i(C)$.
        \item Let $\mathcal{Q} \leftarrow \mathcal{Q} \cup Q_i$ and $v_i = \argmin_{v \in \mathcal{Q}} f(v)$.
    \end{enumerate}

\item \textbf{Round $t$:}
    \begin{enumerate}
        \item Let $\mathcal{A}_{t-1} \subseteq \mathcal{C}_{t-1}$ be the set of final (level $t-1$) components adjacent to $v_{t-1}$.
        \item Query all vertices in these components: $Q_t = \bigcup_{C \in \mathcal{A}_{t-1}} C$.
    \end{enumerate}

\item \emph{Output:} Let $\mathcal{Q} = \mathcal{Q} \cup Q_t$ be the set of all queried vertices. Output $v^* = \argmin_{v \in \mathcal{Q}} f(v)$.
\end{itemize}

\paragraph{Correctness.}
It suffices to show that the neighbors of $v^*$ have been queried: $N(v^*) \subseteq \mathcal{Q}$. Let $u$ be an arbitrary neighbor of $v^*$. We must show that $u \in \mathcal{Q}$.
Let $m = Level(v^*)$ and $j = Level(u)$.

\emph{Case 1: $m=t$ ($v^*$ is in a final component).}
Let $C^* \in \mathcal{C}_{t-1}$ be the final component containing $v^*$.
Let $C_0=V$. The hierarchical decomposition ensures that for $C^*$, there is a unique sequence of ``ancestor components'' $(C_1,  \ldots, C_{t-1}=C^*)$ such that $C_i \in \mathcal{C}_i$ and $C_i \subseteq C_{i-1}$ for $i \ge 1$. Intuitively, $C_i$ is the specific level-$i$ component that contains $v^*$. Let the set of vertices in the \emph{associated separators} for this lineage be
$S^* := S_1 \cup \bigl(\bigcup_{i=2}^{t-1} \sigma_i(C_{i-1})\bigr)\,.$

We first show that $u \in C^* \cup S^*$.
By construction, $C_1$ is a component of $G[V \setminus S_1]$ and  $C_i$ is a  component of $G[C_{i-1} \setminus \sigma_i(C_{i-1})]$ for $i\ge 2$.

If $u \in C^*$, the condition holds. If $u \notin C^*$, there must be some level $i \ge 1$ where $u \in C_{i-1}$ but $u \notin C_i$. Since there is an edge $(v^*, u)$ and $v^* \in C_i$, the definition of the component $C_i$ implies that
 $u \in \sigma_i(C_{i-1})$ if $i \geq 2$ or  $u \in S_1$ if $i=1$.
Thus $u \in S^*$.

Next, to show $u \in \mathcal{Q}$, it suffices to show that $C^* \cup S^* \subseteq \mathcal{Q}$.

\begin{enumerate}
    \item \emph{Querying $C^*$:} Since Level$(v^*)=t$, vertex $v^*$ is not in any separator (which are queried in rounds 1 to $t-1$). As $v^* \in \mathcal{Q}$, it must have been queried in round $t$. This implies $C^* \in \mathcal{A}_{t-1}$, meaning the entire component $C^*$ was queried ($C^* \subseteq Q_t \subseteq \mathcal{Q}$).

    \item \emph{Querying $S^*$:} $S_1$ is queried in Round 1 ($S_1 \subseteq \mathcal{Q}$). For $i \ge 2$, consider the ancestor component $C_{i-1}$. We know $v^* \in C^* \subseteq C_{i-1}$. By Lemma~\ref{claim:non_exploration},  if $C_{i-1}$ were not in $\mathcal{A}_{i-1}$, then no vertex inside $C_{i-1}$ (including $v^*$) would be queried. Since $v^*$ is queried, it follows that $C_{i-1} \in \mathcal{A}_{i-1}$, and so the separator $\sigma_i(C_{i-1})$ was queried in Round $i$ (i.e. $\sigma_i(C_{i-1}) \subseteq \mathcal{Q}$).
\end{enumerate}
Since $u \in C^* \cup S^*$ and $C^* \cup S^* \subseteq \mathcal{Q}$, we have $u \in \mathcal{Q}$.

\emph{Case 2: $m < t$ ($v^*$ is in a separator).}
Let $\mathcal{Q}_{t-1}$ be the set of vertices queried up to the end of round $t-1$. We have $v^* \in \mathcal{Q}_{t-1} \subseteq \mathcal{Q}$. Since $v^* = \argmin_{v \in \mathcal{Q}} f(v)$, the tie-breaking rule of the algorithm implies $v^* = \argmin_{v \in \mathcal{Q}_{t-1}} f(v)$. By definition, this means $v^* = v_{t-1}$.

Recall $j = \text{Level}(u)$, where $u$ is the neighbor of $v^*$ we aim to show is in $\mathcal{Q}$. We consider a few cases:

\begin{description}
\item \textit{Case 2.1: $j=t$ ($u$ is in a final component $C$).}
We have $C \in \mathcal{C}_{t-1}$.
By the definition of round $t$, the algorithm identifies and queries the set $\mathcal{A}_{t-1}$ of components adjacent to $v_{t-1}$. Since $u \in C$ neighbors $v_{t-1}=v^*$, we have $C \in \mathcal{A}_{t-1}$. Thus  component $C$ is queried, so  $u \in \mathcal{Q}$.

\item \emph{Case 2.2: $j < t$ ($u$ is in a separator).}
Thus $u$ belongs to some level-$j$ separator, so there exists a component $C_{j-1} \in \mathcal{C}_{j-1}$ with $u \in S_j(C_{j-1})$ (where  $C_0=V$ and $S_1(C_0)=S_1$ if $j=1$).
\noindent {\emph{Case 2.2.1: $j=m$.}}
Since $u$ and $v$ are adjacent and at the same level $j$, they  belong to the same separator ($S_j(C_{j-1})$) by the hierarchical decomposition.
Since $v^*$ was queried (in round $m=j$), the entire separator $S_j(C_{j-1})$ must have been queried in round $j$, so   $u \in Q_j \subseteq \mathcal{Q}$.

\noindent {\emph{Case 2.2.2: $j \geq  m+1$} ($u$ is deeper in the hierarchy than $v^*$).}
As $v^*$ was queried in round $m$, we show the running minimum stabilizes at $v^*$ from round $m$ onwards. Formally, we show by induction that $v_i=v^*$ for $i=m, \ldots, t-1$.

Base Case ($i=m$): Since  $v^* = \argmin_{\mathcal{Q}} f$, vertex $v^*$ was queried in round $m$, and the algorithm uses lexicographic tie-breaking, we have  $v^* = \argmin_{\mathcal{Q}_m} f = v_m=v^*$.

Inductive Step: Assume $v_{i-1}=v^*$.  Since $i \ge m$, we have $v^* \in \mathcal{Q}_m \subseteq \mathcal{Q}_i  \subseteq \mathcal{Q}$. Since $v^* = \argmin_{\mathcal{Q}} f$ and $v^* \in \mathcal{Q}_i$, we get $v^* = \argmin_{\mathcal{Q}_i} f$. Thus $v_i=v^*$, completing the induction.

We know $u \in S_j(C_{j-1}) \subseteq C_{j-1}$. Since $u$ is a neighbor of $v^*=v_{j-1}$, we have $C_{j-1} \in \mathcal{A}_{j-1}$, so the separator $S_j(C_{j-1})$ is queried in round $j$. Thus $u \in Q_j \subseteq \mathcal{Q}$.

\noindent {\emph{Case 2.2.3: $j \leq m -1$} ($u$ is shallower in the hierarchy than $v^*$).} Since
$u \in S_j(C_{j-1})$, we have  $v^* \in C_{j-1}$.
If $j=1$, then $u \in S_1$ and $S_1$ is queried in round 1, so $u \in \mathcal{Q}$.
Else
 $j\ge 2$. Since $v^* \in C_{j-1}$ and $v^*$ is queried,  Lemma~\ref{claim:non_exploration} (with $i=j$ and $C = C_{j-1}$) implies $C_{j-1} \in \mathcal{A}_{j-1}$. Thus the separator $S_j(C_{j-1})$ containing $u$ was queried in round $j$, so $u \in Q_j \subseteq \mathcal{Q}$.
\end{description}

This completes the correctness argument.

\paragraph{Query Complexity Analysis.}
The total number of queries is
    $\sum_{i=1}^{t} |Q_i| \,.$
By Lemma~\ref{lem:shattering},
\begin{align} \label{eq:separator_bounds_at_each_level}
|S_1| < \frac{3sn}{K_1} \; \; \mbox{ and } \; \; |\sigma_i(C_{i-1})| <  \frac{3s |C_{i-1}|}{K_i} \le \frac{3s K_{i-1}}{K_i}
\; \; \forall i \in \{2, \ldots, t-1\}\,.
\end{align}
Using \eqref{eq:separator_bounds_at_each_level} we bound the number of queries in each round as follows:
\begin{enumerate}[(i)]
\item Round 1: Since $Q_1 = S_1$, we get $|Q_1| < {3sn}/{K_1}$.
\item Rounds $2 \leq i \leq t-1$: We have  $|\mathcal{A}_{i-1}| \le \Delta$ since $v_{i-1}$ has at most $\Delta$ neighbors. Thus  $|Q_i| < \Delta \cdot {3s \cdot K_{i-1}}/{K_i}$.
\item Round $t$:  $|Q_t| \le \Delta \cdot  K_{t-1}$.
\end{enumerate}

Define  $f : \mathbb{R}^{t-1} \to \mathbb{R}$ such that for each $\vec{K} = (K_1, \ldots, K_{t-1}) \in \mathbb{R}^{t-1}$,
\begin{align} \label{eq:ub_k_rounds}
f(\vec{K}) := \frac{3sn}{K_1} + \left( \sum_{i=2}^{t-1}  \frac{3s\Delta \cdot K_{i-1}}{K_i}  \right)+ \Delta \cdot  K_{t-1}.
\end{align}

The  bounds in (i-iii) imply that the total query complexity is at most $f(\vec{K})$. We first analyze the continuous relaxation of the problem where $\vec{K} \in \mathbb{R}_{+}^{t-1}$.
Applying the AM-GM inequality yields
\begin{align}
    f(\vec{K}) \geq t \sqrt[t]{\left( \frac{3sn}{K_1} \right) \cdot \left( \prod_{i=2}^{t-1} \frac{3s \Delta \cdot K_{i-1}}{K_i} \right) \cdot \left( \Delta \cdot K_{t-1} \right)} = k n^{\frac{1}{t}} \left( 3s \Delta\right)^{\frac{t-1}{t}} \,.
\end{align}

The minimum is obtained when the terms in the arithmetic mean are equal, corresponding to the solution $\vec{K}^* = (K_1^*, \ldots, K_{t-1}^*)$ defined by:
 $K_{i}^* = (3s)^{\frac{i}{t}} \cdot \left(\frac{n}{\Delta}\right)^{1-\frac{i}{t}} \; \; \mbox{for } i \in [t-1]\,.$
Then
\begin{align} \label{eq:f_K_star}
    f(\vec{K}^*) & = t \cdot (3s\Delta)^{1-\frac{1}{t}} \cdot n^{\frac{1}{t}} \,.
\end{align}
\emph{Case (a): $3s\Delta < n$.} Let $\widehat{K}_i = \lceil K_i^* \rceil$  $\forall i \in [t-1]$. Using the inequality $K_i^* \leq \widehat{K}_i < K_i^* + 1$, we get
\begin{align} \label{eq:ub_f_K_hat}
    f(\widehat{\vec{K}}) & =  \frac{3sn}{\widehat{K}_1} + \left( \sum_{i=2}^{t-1}  \frac{3s\Delta \cdot \widehat{K}_{i-1}}{\widehat{K}_i} \right) + \Delta \cdot  \widehat{K}_{t-1} \notag \\
    & < \frac{3sn}{{K}_1^*} + \left( \sum_{i=2}^{t-1}  \frac{3s\Delta \cdot ({K}^*_{i-1} +1) }{{K}^*_i}  \right) + \Delta \cdot ( {K}^*_{t-1} + 1)  = f(\vec{K}^*) + 3s\Delta \left( \sum_{i=2}^{t-1} \frac{1}{K_i^*} \right) + \Delta  \,.
\end{align}

Using the definition of $\vec{K}^*$  in \eqref{eq:ub_f_K_hat}, we obtain:
\begin{align} \label{eq:bound_additive_error_from_rounding}
f(\widehat{\vec{K}}) & < f(\vec{K}^*) + 3s \Delta \left(\sum_{i=2}^{t-1} \frac{1}{(3s)^{\frac{i}{t}}  \left(\frac{n}{\Delta}\right)^{1-\frac{i}{t}}} \right) + \Delta  
 = f(\vec{K}^*) + \Delta 
\left( \frac{ 1 - \left(\frac{n}{3s\Delta} \right)^{\frac{2}{t} - 1}}{ \left(\frac{n}{3s\Delta} \right)^{\frac{1}{t}}  - 1} \right) + \Delta 
\leq f(\vec{K}^*) + t \Delta,
\end{align}
where the last inequality in \eqref{eq:bound_additive_error_from_rounding} used  Lemma~\ref{lem:ub_t_minus_one} with $x=n/(3s\Delta)$.
Using \eqref{eq:f_K_star}  in  \eqref{eq:bound_additive_error_from_rounding}, we get  
\begin{align}\label{eq:ub_K_widetilde_intermediate}
    f(\widehat{\vec{K}}) & <  f(\vec{K}^*) + t \Delta = t \cdot (3s\Delta)^{1-\frac{1}{t}} \cdot n^{\frac{1}{t}}  + t \Delta  \,.
\end{align}
Since $s \geq 1$ and $n \geq \Delta$, we have
$    (3s\Delta)^{1-\frac{1}{t}} \cdot n^{\frac{1}{t}} \geq (3 \Delta)^{1-\frac{1}{t}} \cdot \Delta^{\frac{1}{t}} = 3^{1-\frac{1}{t}} \cdot \Delta $. Thus we can bound the right hand side of  \eqref{eq:ub_K_widetilde_intermediate} as follows
\begin{align}
    f(\widehat{\vec{K}}) & <  t \cdot (s\Delta)^{1-\frac{1}{t}} \cdot n^{\frac{1}{t}} \left[1 + 3^{1-\frac{1}{t}} \right] < 4t   (s\Delta)^{1-\frac{1}{t}}  n^{\frac{1}{t}} \,.
\end{align}
Thus running the algorithm with parameters $\widehat{\vec{K}}$ bounds the number of queries to $4t   (s\Delta)^{1-\frac{1}{t}}  n^{\frac{1}{t}}$.

\emph{Case (b): $3s\Delta \geq  n$.}
In this case, the bound derived from this optimization strategy is worse than the trivial complexity $n$.

Combining both cases, the  query complexity is bounded by $\min(n, 4t   (s\Delta)^{1-\frac{1}{t}}  n^{\frac{1}{t}})$ as required.
\end{proof}

\begin{lemma}[Non-Exploration] \label{claim:non_exploration}
    In the setting of Theorem~\ref{thm:k_rounds}, let $i \in \{2, \ldots, t\}$ and $C \in \mathcal{C}_{i-1}$. If the running minimum $v_{i-1}$ has no neighbors in $C$, then no vertex in $C$ is ever queried by the algorithm.
\end{lemma}
The proof of the lemma is in Appendix~\ref{app:algorithms}, together with  other proofs omitted from the main text.

\subsection{Randomized Algorithm} \label{sec:randomized}

We also consider a randomized algorithm, which is  a parallelized version of classical steepest descent with a warm start algorithm~\cite{aldous1983minimization}.

To rigorously handle potential ties in function values $f: V \to \mathbb{R}$ and ensure a well-defined search path, we define a strict total order $\prec$ of the vertices using their values and lexicographic tie-breaking.
    For each $u,v \in V$, we have $u \prec v$ if:
    \begin{itemize}
        \item $f(u) < f(v)$; or
        \item $f(u) = f(v)$ and the index of $u$ is smaller than the index of $v$.
    \end{itemize}
    The rank of $v \in V$, denoted $\Rank(v) \in [n]$, is its position in this total order.

    The \emph{steepest descent path} from $v$ is the sequence of vertices $v_0, v_1, \ldots, v_k$ such that $v_0 = v$ and
    $
        v_{i+1} = \argmin_{u \in N(v_i)} \{ u \mid u \prec v_i \}
    $ for each $i < k$.
    The path terminates at $v_k$ when $v_k$ is a local minimum with respect to $\prec$. The length of this path is denoted $\mathcal{L}(v) = k \leq \Rank(v)-1$. A local minimum with respect to  $\prec$ is also a local minimum with respect to $f$.

To bound the query complexity, we first quantify the quality of the warm start. \footnote{In Lemma~\ref{lem:expected_rank_replacement} we use sampling with replacement as it does not require shared state among different processors.}

\begin{lemma} \label{lem:expected_rank_replacement}
Let $Q$ be a multiset of $q$ vertices sampled uniformly at random with replacement from $V$. Let $v_{\min}$ be the minimum vertex in $Q$ with respect to $\prec$. Then $\E[\mathcal{L}(v_{\min})] < \frac{n}{q+1}$.
\end{lemma}

For a vertex $v \in V$ and  $\rho \in \mathbb{R}_+$, let the ball of radius $\rho$ centered at $v$ be
\begin{align}
B(v; \rho) = \{ u \in V \mid dist(v, u) \le \rho \} \,.
\end{align}

The next lemma bounds the size of the ball centered at $v$ with radius $\rho$.
\begin{lemma} \label{lem:ball_size_rigorous}
    Let $G = (V,E)$ be a graph with maximum degree $\Delta \ge 2$. Let $v \in V$ and  $\rho \in \mathbb{N}^*$.  
      If $\Delta=2$, then $|B(v; \rho)| \le 2\rho+1$.
        Else if $\Delta \ge 3$, then $|B(v; \rho)| < \frac{\Delta}{\Delta-2}(\Delta-1)^\rho$.
\end{lemma}

\medskip 

\refstepcounter{algorithm} \label{alg:wssd}
\noindent\textbf{Algorithm \thealgorithm: Parallel Steepest Descent with a Warm Start.}

The algorithm operates in $t$ rounds and is parameterized by a sample size $q_1$ and a search radius $r$.  It uses one round for sampling and $t-1$ rounds for parallel search.
 By querying a ball of radius $r+1$, the algorithm can simulate $r$ steps of the steepest descent path in a single round.

 \paragraph{Round 1 (Warm Start).} Sample a multiset $Q_1$ of $q_1$ vertices chosen uniformly at random with replacement from $V$  and query them. Let $v^{(0)}$ be the minimal vertex in $Q_1$ with respect to $\prec$.

\paragraph{Rounds 2 to $t$ (Parallel Search).}
For each step $i = 1, \dots, t-1$:
\begin{enumerate}
    \item Query all vertices in the ball $B(v^{(i-1)}; r+1)$. Trace the steepest descent path starting from $v^{(i-1)}$ using the queried values.
    \item If the path terminates at a local minimum $v^*$ with $dist(v^*, v^{(i-1)})\leq r$,  output $v^*$ and halt.
     Else, let $v^{(i)}$ be the vertex with $dist(v^{(i-1)}, v^{(i)})= r$ along the steepest descent path.
\end{enumerate}
If no local minimum was found by the end of $t$ rounds, output ``Failure''.

The next figure illustrates a step of the parallel search strategy of the algorithm, which queries a ball around $v^{(0)}$ to identify the initial segment of the steepest descent path.

\definecolor{deepblue}{rgb}{0.1, 0.2, 0.4}
\definecolor{oceanblue}{rgb}{0.0, 0.3, 0.6}
\definecolor{manifoldbase}{rgb}{0.85, 0.88, 0.95}

\definecolor{crimson}{rgb}{0.8, 0.0, 0.0} 
\definecolor{faintred}{rgb}{1.0, 0.8, 0.8}
\usetikzlibrary{calc} 

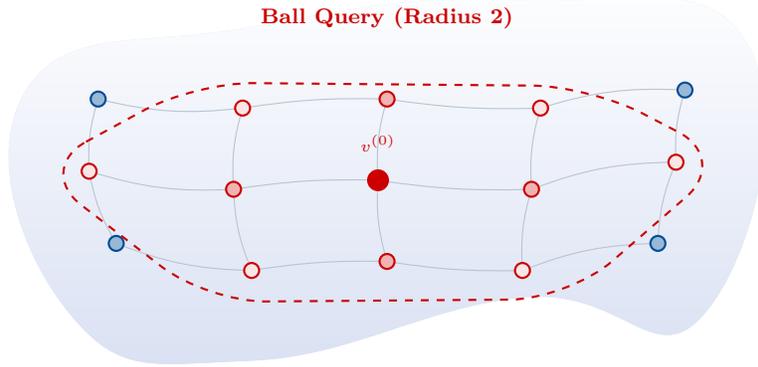
\begin{figure}[h!]
    \centering 
\begin{tikzpicture}[
    scale=1.2,
    vertex/.style={circle, draw=deepblue!40, fill=white, inner sep=1.5pt, thin},
    sample/.style={circle, draw=oceanblue, fill=oceanblue!40, inner sep=2pt, thick},
    searchnode/.style={circle, draw=crimson, fill=crimson!20, inner sep=2pt, thick},
    centernode/.style={circle, draw=crimson, fill=crimson, inner sep=2.5pt, very thick, text=white},
    edge/.style={thin, deepblue!30},
    pathedge/.style={thick, crimson, ->, >=stealth}
] 

    \begin{scope}[on background layer]
        \draw[draw=none, fill=manifoldbase, shade, top color=white, bottom color=manifoldbase] 
            plot [smooth cycle, tension=0.8] coordinates {
                (0,0) (2,-0.5) (5,0.2) (7,0) 
                (7.5,3) (5,3.5) (2,3.2) (-0.5,2.5)
            };
    \end{scope}

    \coordinate (n11) at (0.5, 0.8);  \coordinate (n12) at (2.0, 0.5);  \coordinate (n13) at (3.5, 0.6); \coordinate (n14) at (5.0, 0.5); \coordinate (n15) at (6.5, 0.8);
    \coordinate (n21) at (0.2, 1.6);  \coordinate (n22) at (1.8, 1.4);  \coordinate (n23) at (3.4, 1.5); \coordinate (n24) at (5.1, 1.4); \coordinate (n25) at (6.7, 1.7);
    \coordinate (n31) at (0.3, 2.4);  \coordinate (n32) at (1.9, 2.3);  \coordinate (n33) at (3.5, 2.4); \coordinate (n34) at (5.2, 2.3); \coordinate (n35) at (6.8, 2.5);

    \foreach \r in {1,2,3} {
        \draw[edge] (n\r1) to[bend right=10] (n\r2) to[bend left=5] (n\r3) to[bend right=5] (n\r4) to[bend left=10] (n\r5);
    }
    \foreach \c in {1,2,3,4,5} {
        \draw[edge] (n1\c) to[bend left=10] (n2\c) to[bend left=10] (n3\c);
    }

    \foreach \r in {1,2,3} \foreach \c in {1,2,3,4,5} {
        \node[vertex] at (n\r\c) {};
    }

     
    \foreach \x in {n21, n25, n12, n14, n32, n34} {
        \node[searchnode, fill=crimson!10] at (\x) {}; 
    }
    \foreach \x in {n22, n24, n13, n33} {
        \node[searchnode, fill=crimson!30] at (\x) {};
    }
    \node[centernode] (start) at (n23) {};

    \node[sample] at (n11) {};
    \node[sample] at (n31) {};
    \node[sample] at (n15) {};
    \node[sample] at (n35) {};

     
    \begin{scope}[crimson, dashed, thick, rounded corners=22pt]
        \draw ($(n21)!-0.6cm!(n23)$) -- 
              ($(n32)!-0.6cm!(n23)$) -- 
              ($(n34)!-0.6cm!(n23)$) -- 
              ($(n25)!-0.6cm!(n23)$) -- 
              ($(n14)!-0.6cm!(n23)$) -- 
              ($(n12)!-0.6cm!(n23)$) -- cycle;
    \end{scope}

    \node[font=\scriptsize\bfseries, crimson] at (3.5, 3.3) {Ball Query (Radius 2)};

    \node[above=3pt of start, font=\tiny, crimson] {\textbf{$v^{(0)}$}};

\end{tikzpicture}
\caption{Parallel search from $v^{(0)}$.}
\end{figure}

Next we quantify the performance of this algorithm.

\noindent \textbf{Proposition~\ref{thm:t_rounds_const_deg_generalized}.}
\emph{Let $G=(V,E)$ be a graph with $n$ vertices and maximum degree $\Delta$. The randomized query complexity of finding a local minimum in $t \ge 2$ rounds is $O(\sqrt{n} + t)$ when $\Delta \leq 2$ and $O\bigl(\frac{n}{t \cdot \log_\Delta n} + t \Delta^2 \sqrt{n}\bigr)$ when $\Delta \geq 3$.}

If $\Delta \geq 3$ and $t$ are constants, then the randomized query complexity is $ O(n/\log n)$ even in two rounds. Algorithm~\ref{alg:wssd} is effective for bounded-degree graphs with high expansion (e.g., expanders where $s=\Theta(n)$), where the deterministic bound of Theorem~\ref{thm:k_rounds} is $O(n)$. On graphs with slow expansion (e.g.  grids), specialized algorithms perform better (see, e.g.,~\cite{branzei2022query}).

\section{Lower Bounds} \label{sec:lower_bounds}

 Let $T=(V, E_T)$ be an arbitrary fixed spanning tree of the connected graph $G=(V,E)$, rooted at some vertex $r$.
For $u,v \in V$, let $dist_T(u,v)$ be the distance between $u$ and $v$ in $T$.

Let $Anc_T(v)$ denote the set of ancestors of $v$ in $T$ (the vertices on the unique path from $r$ to $v$ in $T$, including $r$ and $v$). We write $u \preceq_T v$ if $u \in Anc_T(v)$.

For each node $x \in V$, let $\mathcal{T}(x)$ denote the subtree of
$T$ rooted at $x$. Formally, the vertex set of $\mathcal{T}(x)$ is the set of descendants of $x$, i.e., $\{v \in V \mid x \preceq_T v\}$.

\begin{definition}[The Family $\mathcal{F}$ and Distribution $\mathcal{D}$.] \label{def:distribution_D}
Given the spanning tree $T$ of $G$ rooted at $r$, define for
 every vertex $v \in V$ the function $f_{v}: V \to \mathbb{Z}$:
\begin{align}
     f_{v}(x) = \begin{cases} -dist_T(r,x) & \text{if } x \preceq_T v \\ \phantom{-}dist_T(r,x) & \text{otherwise.} \end{cases}
\end{align}
 Let $\mathcal{D}$ be the uniform distribution over $\mathcal{F} = \{f_v \mid v \in V\}$.    
\end{definition}

\begin{remark} \label{rmk:unique_local_min_in_F}
     In Definition~\ref{def:distribution_D}, vertex  $v$ is the unique local  minimum of $f_v$ in the  tree $T$. Moreover, since adding edges to a graph cannot turn a node that is not a local minimum into one that is, the function $f_v$  has a unique local minimum in $G$.
For input distribution $\mathcal{D}$, the target local minimum is a random variable chosen uniformly at random from $V$.
\end{remark}

\paragraph{History and Candidate Set.} Let $\mathcal{A}$ be a deterministic  algorithm that runs in $t$ rounds. A history of $\mathcal{A}$ at the end of round $i$ on input $f \in \mathcal{F}$ represents the sequence of queries issued and answers observed by $\mathcal{A}$ on this input until the end of round $i$.

Let $\mathcal{H}_i$ denote the set of histories reachable by $\mathcal{A}$ at the end of round $i$ on inputs from $\mathcal{F}$.
Given a history $H \in \mathcal{H}_i$, the \emph{candidate set}  $\mathcal{C}(H)$ consists of the vertices that could still be local minima given this history.  We denote it by
\begin{align}
\mathcal{C}(H) = \{v \in V \mid  \text{input } f_v \text{ generates history } H \} \,.
\end{align}

\paragraph{Signatures.} We define the concept of a signature, which captures the information revealed about the ancestry of a vertex based on a set of queries in this construction.

\begin{definition}[Signature]
    Let $U,Q \subseteq V$ be arbitrary sets of vertices.
    For a vertex $u \in U$, the \emph{signature of $u$ with respect to $Q$} is defined as:
    \[ S_Q(u) := Q \cap Anc_T(u) \,. \]
    Let $\mathcal{S}_Q(U) = \bigcup_{v \in U} \{S(v)\}$ be the family of distinct signatures of vertices in $U$ with respect to $Q$.
\end{definition}

Figure~\ref{fig:example_graph_with_spanning_tree} shows an example graph $G$ (left) with a spanning tree (right) and the vertices queried by an algorithm in round 1.

\begin{figure}[h!]
    \centering
    \begin{minipage}[b]{0.48\textwidth}
        \centering
        \resizebox{\linewidth}{!}{%
            \begin{forest}
                for tree={
                    circle,
                    draw,
                    fill=black,
                    inner sep=1.5pt,
                    s sep=12mm,
                    l sep=15mm,
                    font=\small,
                    edge={thick, black},
                    delay={content=\phantom{5}}
                },
                tikz={
                    \draw (n4) to[bend left] (n8);
                    \draw (n8) to[bend left] (n10);
                    \draw (n5) to[bend left] (n11);
                    \draw (n6) to[bend left] (n12);
                    \draw (n11) to[bend left] (n13);
                    \draw (n13) to[bend left] (n16);
                    \draw (n2) to[bend left] (n3);
                    \draw (n4) to[bend left] (n7);
                    \draw (n6) to[bend left] (n11);
                    \draw (n9) to[bend left] (n10);
                    \draw (n13) to[bend left] (n14);
                    \draw (n15) to[bend left] (n16);
                }
                [1, name=n1, label=above:1
                    [2, name=n2, label={[label distance=2mm]left:2}
                        [4, name=n4, label={[label distance=0.5mm]left:4}
                            [5, name=n5, label=below:5]
                            [6, name=n6, label=below:6]
                        ]
                        [7, name=n7, label={[label distance=0.5mm]left:7}]
                    ]
                    [3, name=n3, label={[label distance=2mm]left:3}
                        [8, name=n8, label={[label distance=0.5mm]left:8}
                            [11, name=n11, label=below:11]
                            [12, name=n12, label=below:12]
                            [13, name=n13, label=below:13]
                            [14, name=n14, label=below:14]
                        ]
                        [9, name=n9, label={[label distance=0.5mm]left:9}]
                        [10, name=n10, label={[label distance=1mm]left:10}
                            [15, name=n15, label=below:15]
                            [16, name=n16, label=below:16]
                        ]
                    ]
                ]
            \end{forest}%
        }
        \label{fig:example_graph}
    \end{minipage}
    \hfill 
    \hfill
    \begin{minipage}[b]{0.48\textwidth}
        \centering
        \resizebox{\linewidth}{!}{%
            \begin{forest}
                for tree={
                    circle,
                    draw,
                    fill=black,
                    inner sep=1.5pt,
                    s sep=12mm,
                    l sep=15mm,
                    font=\small,
                    edge={thick, black},
                    delay={content=\phantom{5}}
                },
                tikz={
                    \node[querycircle, fit=(n2)] {};
                    \node[querycircle, fit=(n3)] {};
                    \node[querycircle, fit=(n10)] {};
                }
                [1, name=n1, label=above:1
                    [2, name=n2, label={[label distance=2mm]left:2}
                        [4, name=n4, label={[label distance=0.5mm]left:4}
                            [5, name=n5, label=below:5]
                            [6, name=n6, label=below:6]
                        ]
                        [7, name=n7, label={[label distance=0.5mm]left:7}]
                    ]
                    [3, name=n3, label={[label distance=2mm]left:3}
                        [8, name=n8, label={[label distance=0.5mm]left:8}
                            [11, name=n11, label=below:11]
                            [12, name=n12, label=below:12]
                            [13, name=n13, label=below:13]
                            [14, name=n14, label=below:14]
                        ]
                        [9, name=n9, label={[label distance=0.5mm]left:9}]
                        [10, name=n10, label={[label distance=1mm]left:10}
                            [15, name=n15, label=below:15]
                            [16, name=n16, label=below:16]
                        ]
                    ]
                ]
            \end{forest}%
        }
        \label{fig:spanning_tree_and_round_one_queries}
    \end{minipage}
    \caption{The left figure shows an example of a graph $G$. A spanning  tree of $G$ rooted at vertex 1 is shown on the right, with the set of vertices queried by an algorithm in round 1 circled in red.}
\label{fig:example_graph_with_spanning_tree}
\end{figure}
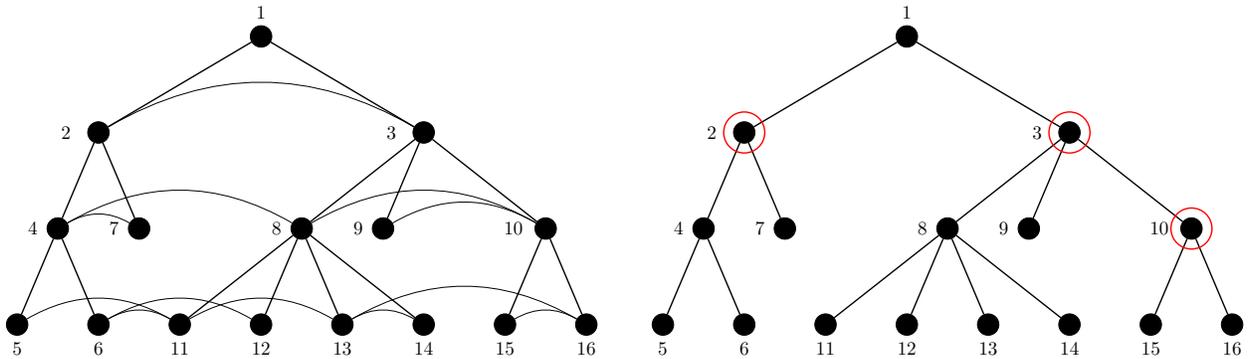

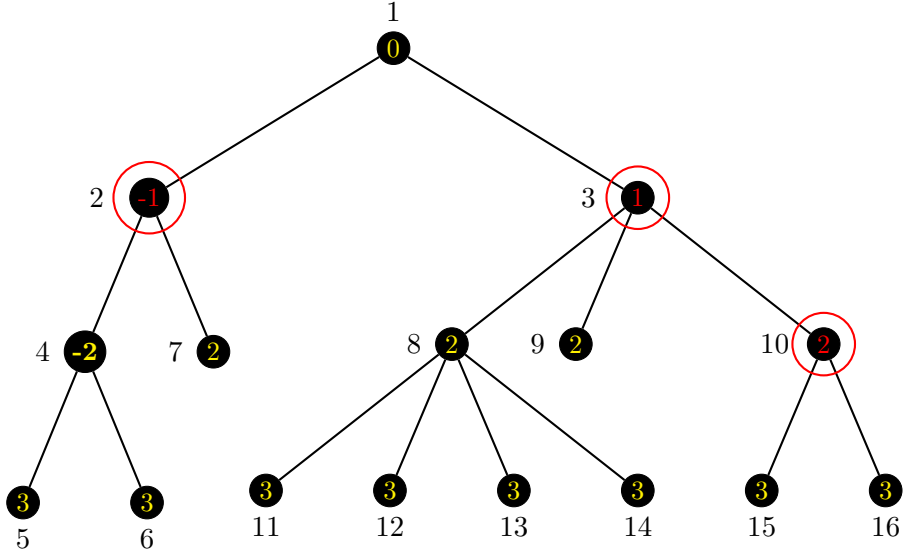
\begin{figure}[h!]
    \centering
    \begin{forest}
        for tree={
            circle,
            draw,
            fill=black,
            text=yellow, 
            inner sep=1.5pt,
            s sep=12mm,
            l sep=15mm,
            font=\small,
            edge={thick, black},
        },
        tikz={
            \node[querycircle, fit=(n2)] {};
            \node[querycircle, fit=(n3)] {};
            \node[querycircle, fit=(n10)] {};
        }
        [0, name=n1, label=above:1
            [-1, name=n2, text=red, label={[label distance=2mm]left:2}
                [-2, name=n4, font=\bfseries\small, label={[label distance=0.5mm]left:4}
                    [3, name=n5, label=below:5]
                    [3, name=n6, label=below:6]
                ]
                [2, name=n7, label={[label distance=0.5mm]left:7}]
            ]
            [1, name=n3, text=red, label={[label distance=2mm]left:3}
                [2, name=n8, label={[label distance=0.5mm]left:8}
                    [3, name=n11, label=below:11]
                    [3, name=n12, label=below:12]
                    [3, name=n13, label=below:13]
                    [3, name=n14, label=below:14]
                ]
                [2, name=n9, label={[label distance=0.5mm]left:9}]
                [2, name=n10, text=red, label={[label distance=1mm]left:10}
                    [3, name=n15, label=below:15]
                    [3, name=n16, label=below:16]
                ]
            ]
        ]
    \end{forest}
    \caption{Spanning tree with the round 1 queries circled in red. The value at each node is also shown, for the case  where the target local minimum is vertex $4$ (i.e. the input function is the function $f_4 \in \mathcal{F}$).}
    \label{fig:spanning_tree_and_round_one_queries_and_function_values}
\end{figure}

Suppose in round 1 the algorithm queries the set of vertices $Q_1=\{2, 3, 10\}$.
Table~\ref{tab:candidate_sets} lists the resulting candidate sets,  for each individual target local minimum.

\begin{table}[h!]
    \centering
    \caption{For each function $f_u \in \mathcal{F}$ (with target node $u$), the candidate set identified by  the algorithm  at the end of round 1 after querying nodes $\{2,3,10\}$.}
    \label{tab:candidate_sets}
    \renewcommand{\arraystretch}{1.1}
    \begin{tabular}{cc @{\quad} | @{\quad} cc}
        \toprule
        \textbf{Target node} ($u$) & \textbf{Candidate set} & \textbf{Target node} ($u$) & \textbf{Candidate set} \\
        \midrule
        1 & $\{1\}$                        & 9  & $\{3, 8, 9, 11, 12, 13, 14\}$ \\
        2 & $\{2, 4, 5, 6, 7\}$           & 10 & $\{10, 15, 16\}$ \\
        3 & $\{3, 8, 9, 11, 12, 13, 14\}$ & 11 & $\{3, 8, 9, 11, 12, 13, 14\}$ \\
        4 & $\{2, 4, 5, 6, 7\}$           & 12 & $\{3, 8, 9, 11, 12, 13, 14\}$ \\
        5 & $\{2, 4, 5, 6, 7\}$           & 13 & $\{3, 8, 9, 11, 12, 13, 14\}$ \\
        6 & $\{2, 4, 5, 6, 7\}$           & 14 & $\{3, 8, 9, 11, 12, 13, 14\}$ \\
        7 & $\{2, 4, 5, 6, 7\}$           & 15 & $\{10, 15, 16\}$ \\
        8 & $\{3, 8, 9, 11, 12, 13, 14\}$ & 16 & $\{10, 15, 16\}$ \\
        \bottomrule
    \end{tabular}
\end{table}

\subsection{Properties of the Construction}

The following lemma bounds the number of distinct outcomes obtainable from a batch of queries.

\begin{lemma}[Signature Lemma] \label{lem:signatures}
   Let $T$ be a  spanning tree of $G$ rooted at $r$. For all $U,Q \subseteq V$, the number of distinct signatures is bounded by:
    $ | \mathcal{S}_Q(U) | \leq |Q| + 1 \,. $ 
\end{lemma}
\begin{proof}
    Let $\mathcal{S}^* = \mathcal{S}_Q(U) \setminus \{\emptyset\}$ be the set of non-empty signatures. For any $S \in \mathcal{S}^*$, the elements of $S$ lie on a path from the root $r$ and are thus totally ordered by the ancestor relation $\preceq_T$.
   
    We define a map $m: \mathcal{S}^* \to Q$. For each $S \in \mathcal{S}^*$, let $m(S)$ be the unique vertex in $S$ that is farthest from the root $r$ (i.e., $m(S)$ is the deepest node in $S$).

    We show that $m$ is injective. Let $A \in \mathcal{S}^*$ and let $x = m(A)$.
   
    1. ($A \subseteq Anc_T(x) \cap Q$): By definition of $m(A)$, every vertex $y \in A$ must be an ancestor of $x$. Thus $A \subseteq Anc_T(x)$. Since $A \subseteq Q$, we have $A \subseteq Anc_T(x) \cap Q$.

    2. ($Anc_T(x) \cap Q \subseteq A$): Since $A \in \mathcal{S}_Q(U)$, there exists $v \in U$ such that $A = Anc_T(v) \cap Q$. Since $x \in A$, we have $x \in Anc_T(v)$. 
    This implies $Anc_T(x) \subseteq Anc_T(v)$. Intersecting both sides with $Q$ yields $Anc_T(x) \cap Q \subseteq Anc_T(v) \cap Q = A$.
   
    Combining  inclusions (1) and (2), we conclude that $A = Anc_T(x) \cap Q$.

    If $m(A) = m(B) = x$ for some $A, B \in \mathcal{S}^*$, then $A = Anc_T(x) \cap Q = B$. Thus $m$ is injective.
   
    Therefore, $|\mathcal{S}^*| \leq |Q|$. Including the potential empty signature $\emptyset$, we have $|\mathcal{S}_Q(U)| \leq |Q| + 1$.
\end{proof}

 Given a deterministic algorithm $\mathcal{A}$ and $i\in \mathbb{N}$, recall $\mathcal{H}_i$ denotes the set of histories reachable by $\mathcal{A}$ at the end of round $i$ on inputs from $\mathcal{F}$.

For each $H \in \mathcal{H}_i$, let $\mathcal{Q}(H) \subseteq V$ denote the vertices queried in $H$,  $\mathcal{Q}^-(H)=\{x \in Q \mid f(x)<0\}$, $\mathcal{Q}^+(H)=\{x \in Q \mid f(x)>0\}$, and $r_H$  the unique deepest node in $\mathcal{Q}^- \cup \{r\}$.

\begin{lemma} \label{lem:candidate_set_characterization}
Let $\mathcal{A}$ be a deterministic algorithm and $\mathcal{H}_i$  the set of  histories reachable by $\mathcal{A}$ at the end of round $i \in \mathbb{N}$ on inputs from $\mathcal{F}$.
Then:
\begin{enumerate}[(a)]
    \item For each history $H \in \mathcal{H}_i$, the candidate set $\mathcal{C}(H)$ at the end of round $i$ is the set of vertices of the tree
    \begin{align}
    T_i := \mathcal{T}(r_H) \setminus \left( \cup_{x \in \mathcal{Q}^+(H)} \mathcal{T}(x) \right) \,.
    \end{align}
    The only vertex possibly queried at the end of round $i$ in $\mathcal{C}(H)$ is $r_H$ ($\mathcal{C}(H) \cap \mathcal{Q}(H) \subseteq \{r_H\}$).
    \item The collection of non-empty candidate sets $\mathcal{R}_i = \bigcup_{H \in \mathcal{H}{_i}} \{ \mathcal{C}(H) \}$
    forms a partition of $V$ and the map $H \mapsto \mathcal{C}(H)$ is a bijection from the reachable histories $\mathcal{H}_i$ to $\mathcal{R}_i$.
\end{enumerate}
\end{lemma}
An illustration of the partition $\mathcal{R}_1$ induced by the queries $Q_1$ issued by an algorithm in round 1 can be seen in  Figure~\ref{fig:tree_partition_bubbles}.

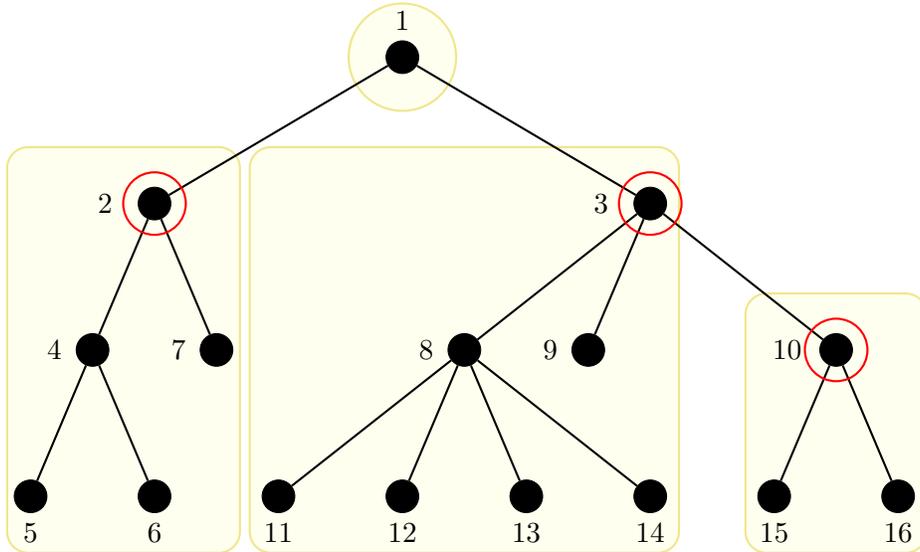
\begin{figure}[h!]
    \centering
    \begin{forest}
        for tree={
            circle,
            draw,
            fill=black,
            inner sep=1.5pt,
            s sep=12mm,      
            l sep=15mm,      
            font=\small,
            edge={thick, black},
            delay={content=\phantom{5}}
        }, 
        tikz={
            \begin{scope}[on background layer]
                \node[bubble, circle, fit=(n1), inner sep=8pt] {};
                \node[bubble, fit=(n2) (n5) (n7), inner ysep=15pt, inner xsep=2.5pt] {};
                \node[bubble, fit=(n3) (n11) (n9), inner ysep=15pt, inner xsep=4.5pt] {};
                \node[bubble, fit=(n10) (n15) (n16),inner ysep=15pt, inner xsep=4.5pt] {};
            \end{scope}
            \node[querycircle, fit=(n2)] {};
            \node[querycircle, fit=(n3)] {};
            \node[querycircle, fit=(n10)] {};
        } 
        [1, name=n1, label=above:1
            [2, name=n2, label={[label distance=2mm]left:2}
                [4, name=n4, label={[label distance=0.5mm]left:4}
                    [5, name=n5, label=below:5]
                    [6, name=n6, label=below:6]
                ]
                [7, name=n7, label={[label distance=0.5mm]left:7}]
            ]
            [3, name=n3, label={[label distance=2mm]left:3}
                [8, name=n8, label={[label distance=0.5mm]left:8}
                    [11, name=n11, label=below:11]
                    [12, name=n12, label=below:12]
                    [13, name=n13, label=below:13]
                    [14, name=n14, label=below:14]
                ]
                [9, name=n9, label={[label distance=0.5mm]left:9}]
                [10, name=n10, label={[label distance=1mm]left:10}
                    [15, name=n15, label=below:15]
                    [16, name=n16, label=below:16]
                ]
            ]
        ]
    \end{forest}
    \caption{Partition of candidate sets (yellow bubbles) with queries $Q_1=\{2, 3, 10\}$ circled in red.}
    \label{fig:tree_partition_bubbles}
\end{figure}

\begin{proof}[Proof of Lemma~\ref{lem:candidate_set_characterization}]
\textbf{Part (a).} Let $H \in \mathcal{H}_i$.
Let $Z \in V$ be the vertex for which $f = f_Z \in \mathcal{F}$.
    We first observe that $\mathcal{Q}^-$ is totally ordered by $\preceq_T$. Let  $v \in V$ be in the candidate set $\mathcal{C}(H)$. Since the function  $f_v$ is consistent with the history $H$, we have:
\begin{itemize}
\item \emph{(C1).} $\forall x \in \mathcal{Q}^-(H)$, $f_v(x) < 0$. This implies $x \preceq_T v$.
\item \emph{(C2).} $\forall x \in \mathcal{Q}^+(H)$, $f_v(x) > 0$. This implies $x \not\preceq_T v$.
\item \emph{(C3).} If $r \in \mathcal{Q}$, $f_v(r)=0$, which is always true.
\end{itemize}

\emph{The root $r_H$ is unique and well-defined.} Since $H$ is reachable, $\mathcal{C}(H) \neq \emptyset$. By (C1), we have  $\mathcal{Q}^-(H) \cup \{r\} \subseteq Anc_T(Z)$. Since $Anc_T(Z)$ is totally ordered, $\mathcal{Q}^-(H) \cup \{r\}$ is totally ordered. Thus $r_H$, the deepest node in $\mathcal{Q}^-(H) \cup \{r\}$, is unique and well-defined. Condition (C1) is equivalent to $r_H \preceq_T v$.

\emph{Structure of $\mathcal{C}(H)$.}
We have $v \in \mathcal{C}(H)$ if and only if $v \in \mathcal{T}(r_H)$ (from C1) and $v \notin \mathcal{T}(x)$ for all $x \in \mathcal{Q}^+(H)$ (from C2).
Thus $\mathcal{C}(H)$ is the set of vertices of the tree $T_i = \mathcal{T}(r_H) \setminus (\cup_{x \in \mathcal{Q}^+(H)} \mathcal{T}(x))$.

We verify that $T_i$ is a connected subtree rooted at $r_H$.
First, we check $r_H \in \mathcal{C}(H)$. We must confirm that for all $x \in \mathcal{Q}^+(H)$, $x \not\preceq_T r_H$. Suppose $x \preceq_T r_H$ for some $x \in \mathcal{Q}^+(H)$. Since $r_H \in \mathcal{Q}^-(H) \cup \{r\}$ and $Z$ is the unique local (and global) minimum, we have $r_H \preceq_T Z$. By transitivity, $x \preceq_T Z$. This implies $f_Z(x) \le 0$. This contradicts $x \in \mathcal{Q}^+(H)$,
since $H$ is the history generated by $f_Z$ which recorded a positive value at $x$. Thus $r_H \in \mathcal{C}(H)$.

Second, we show $T_i$ is connected. For $v \in \mathcal{C}(H)$, consider the path from $r_H$ to $v$. Let $u$ be on this path ($r_H \preceq_T u \preceq_T v$). If $u \notin \mathcal{C}(H)$, then $x \preceq_T u$ for some $x \in \mathcal{Q}^+(H)$. By transitivity, $x \preceq_T v$, contradicting $v \in \mathcal{C}(H)$. Thus the path is entirely in $\mathcal{C}(H)$, proving connectivity.

\emph{Queried vertices in $T_i$.}
If $\mathcal{C}(H) \cap \mathcal{Q} = \emptyset$, then trivially $\mathcal{C}(H) \cap \mathcal{Q} \subseteq \{r_H\}$. Else,
let $y \in \mathcal{C}(H) \cap \mathcal{Q}$.
\begin{itemize}
\item If $y \in \mathcal{Q}^+(H)$, we immediately get a contradiction since $y$ cannot be a candidate.
\item If $y \in \mathcal{Q}^-(H)$, we have $y \preceq_T r_H$ by definition of $r_H$. Since $y \in \mathcal{C}(H) \subseteq \mathcal{T}(r_H)$, we have $r_H \preceq_T y$, so  $y=r_H$.
\item Else, $f(y)=0$, so $y=r =r_H$.
\end{itemize}
Therefore, $\mathcal{C}(H) \cap \mathcal{Q} \subseteq \{r_H\}$.

\textbf{Part (b).} Now we show that $\mathcal{R}_i$ forms a partition of $V$. First we show coverage.
 For all $v \in V$, the function $f_v \in \mathcal{F}$ is a valid input. Since $\mathcal{A}$ is deterministic, $f_v$ generates a unique reachable history $H$. Thus $v \in \mathcal{C}(H)$, so every vertex belongs to some candidate set in $\mathcal{R}_i$.

Second, we show the candidate sets in $\mathcal{R}_i$ are disjoint.  Let $H, \widetilde{H} \in \mathcal{H}_i$ be two distinct reachable histories. 
Since $\mathcal{A}$ is deterministic, the sequence of queries/answers issued is identical in both histories up to the first round where the oracle answers differ at some vertex   $q \in V$.
The oracle response for $q$ is in $\{ -dist_T(r,q), +dist_T(r,q) \}$ for  each input function $f \in \mathcal{F}$. Since the magnitude $dist_T(r,q)$ is fixed by the tree structure, the only information in the response is the sign of $f$ at $q$.
 Since the histories cannot disagree at $r$,  w.l.o.g.  $H$ records a negative sign at $q$ ($f(q) < 0$) and $\widetilde{H}$ records a positive sign ($f(q) > 0$):
\begin{enumerate}
    \item For each candidate $u \in \mathcal{C}(H)$, the function $f_u$ is consistent with $H$, so $f_u(q) < 0$. By  construction of $\mathcal{F}$, this implies $q \in Anc_T(u)$, or equivalently, $u$ is in the subtree $\mathcal{T}(q)$. Hence, $\mathcal{C}(H) \subseteq \mathcal{T}(q)$.
    \item For each candidate $u \in \mathcal{C}(\widetilde{H})$, the function $f_u$ is consistent with $\widetilde{H}$, so $f_u(q) > 0$. By construction of $\mathcal{F}$, this implies $q \notin Anc_T(u)$, meaning $u$ is not in the subtree $\mathcal{T}(q)$. Hence, $\mathcal{C}(\widetilde{H}) \subseteq V \setminus \mathcal{T}(q)$.
\end{enumerate}
Since $\mathcal{T}(q)$ and $V \setminus \mathcal{T}(q)$ are disjoint, we get  $\mathcal{C}(H) \cap \mathcal{C}(\widetilde{H}) = \emptyset$. Thus $\mathcal{R}_i$ is a partition of $V$. It follows immediately that the map $H \mapsto \mathcal{C}(H)$ is a bijection from the reachable histories $\mathcal{H}_i$ to $\mathcal{R}_i$.
\end{proof}

\subsection{Lower Bound for Two Rounds}

We first show a lower bound for two rounds, since it is simpler and illustrates  some of the main ideas that we later build on in the lower bound for $t$ rounds.

\begin{proposition} \label{thm:lb_two_rounds}
    Let $G = (V,E)$ be a connected undirected graph with $n > 1$ vertices. The randomized query complexity of finding a local minimum on $G$ in two rounds is $\Omega(\sqrt{n})$. 
\end{proposition}

\begin{proof}
We show that for each $c \in (1/n, 1]$, the randomized query complexity of finding a local minimum with success probability at least $c$ is no less than $2c\sqrt{n} - 2$. This implies the statement.

We use Yao's Minimax Principle. Consider the probability distribution $\mathcal{D}$ from Definition~\ref{def:distribution_D}. We  lower bound the expected query complexity of any deterministic two-round algorithm $\mathcal{A}$ that succeeds with probability at least $c$ when the input is drawn from $\mathcal{D}$.

By Remark~\ref{rmk:unique_local_min_in_F},
 for each $v \in V$, vertex  $v$ is the unique local  minimum of $f_v$ in  $G$.
 The target local minimum $Z$ is a random variable chosen uniformly at random from $V$.

    \textbf{Round 1:}
    Since $\mathcal{A}$ is deterministic, it selects a fixed set of queries $Q_1 \subseteq V$ in round 1. Let $q_1 = |Q_1|$.
       The algorithm knows the spanning tree $T$ and its root $r$, so a query at a vertex $x$ reveals  $f_{Z}(x)$, or equivalently  whether $x \preceq_T Z$. Thus, the information gained in round 1 is the signature of $Z$ with respect to $Q_1$, namely $S_{Q_1}(Z) = Anc_T(Z) \cap Q_1$.
   
    Let $\mathcal{K} = \bigcup_{v \in V} \{S_{Q_1}(v)\}$
be the family of all possible signatures with respect to $Q_1$. Invoking  Lemma~\ref{lem:signatures} with $U=V$ and $Q=Q_1$ gives $|\mathcal{K}| \leq q_1 + 1$.

For each  signature $\sigma \in \mathcal{K}$, let $C_\sigma$ be the set of  candidates consistent with signature $\sigma$, defined as:
    \[ C_\sigma := \{ v \in V \mid S_{Q_1}(v) = \sigma \} \,. \]
By Lemma~\ref{lem:candidate_set_characterization}, the candidate set $C_\sigma$ is a sub-tree of $T$ such that none of its vertices have been queried except possibly its root.

    \textbf{Round 2:}
    Depending on the  signature $\sigma \in \mathcal{K}$ observed in round 1, the algorithm submits a new batch of queries $Q_{2,\sigma}$ in round 2. Let $q_{2,\sigma} = |Q_{2,\sigma}|$.
    The oracle reveals the signature of $Z$ with respect to this new set: $S_{Q_{2,\sigma}}(Z) = Anc_T(Z) \cap Q_{2, \sigma}$.
   
    Since  $\mathcal{A}$ is deterministic, it  maps the sequence of query results to a single output vertex. Since  $\mathcal{A}$ knows that $Z \in C_\sigma$, it must distinguish the correct solution from other vertices in $C_\sigma$ using the round 2 signature. We say the algorithm succeeds if it outputs   the correct local minimum. If two distinct vertices $u, v \in C_\sigma$ have the same signature in round 2, the algorithm receives identical inputs and produces the same output, succeeding for at most one of them.
    Therefore, the number of vertices $k_\sigma$ in $C_\sigma$ for which the algorithm succeeds is bounded by the number of distinct round 2 signatures generated by $C_\sigma$:
    \begin{align}
    k_\sigma \le \left| \left\{   \bigcup_{v \in C_{\sigma}} \{ S_{Q_{2,\sigma}}(v) \} \right \} \right|  \,.
    \end{align}
    Applying Lemma~\ref{lem:signatures} with candidate set $U=C_\sigma$ and query set $Q=Q_{2,\sigma}$ implies:
    \begin{equation} \label{eq:round2_bound}
        q_{2,\sigma} \geq k_\sigma - 1 \,.
    \end{equation}

\paragraph{Success Probability and Expected Cost.}
We have
\begin{align} \label{eq:Z_in_C_sigma}
    \Pr(Z \in C_{\sigma}) = \frac{|C_{\sigma}|}{n}\,.
\end{align}
At the end of round 1, we know that $Z$ is in $C_{\sigma}$. Moreover  $Z$ is uniformly distributed in $C_{\sigma}$ since the initial distribution was uniform.  Since $\mathcal{A}$  is deterministic,  in round 2 it performs a fixed batch of queries given the round 1 outcome. These queries allow it to distinguish (and thus correctly identify)  $k_\sigma$ distinct vertices.
Then the probability of success given that $Z$ is in $C_{\sigma}$ is  
\begin{align} \label{eq:A_succeeds_given_Z_in_C_sigma}
    \Pr(\mathcal{A} \mbox{ succeeds} \mid Z \in C_{\sigma}) = \frac{k_{\sigma}}{|C_{\sigma}|}\,.
\end{align}
Combining \eqref{eq:Z_in_C_sigma} and \eqref{eq:A_succeeds_given_Z_in_C_sigma}, the overall success probability is
    \begin{align}
    P_{succ} = \Pr(\mathcal{A} \mbox{ succeeds}) = \sum_{\sigma \in \mathcal{K}} \Pr(Z \in C_\sigma) \cdot \Pr(\mathcal{A} \mbox{ succeeds} \mid Z \in C_{\sigma})  
    = \sum_{\sigma \in \mathcal{K}} \frac{k_{\sigma}}{n} \,.
    \end{align}
 
    We assumed $P_{succ} \geq c$, so $\sum_{\sigma \in \mathcal{K}} k_\sigma \geq cn$.

Let $W$ be the random variable for the number of queries issued by the algorithm.
The expected number of queries is the sum of the queries in Round 1 (which are fixed) and the expected number of queries in Round 2 (which depend on the observed signature).

   We have
    \begin{align}
        \mathbb{E}[W] &= q_1 + \sum_{\sigma \in \mathcal{K}} \Pr(Z \in C_\sigma) \cdot q_{2,\sigma}
        = q_1 + \frac{1}{n} \sum_{\sigma \in \mathcal{K}} |C_\sigma| \cdot  q_{2,\sigma} \,.
    \end{align}

    Using $q_{2,\sigma} \geq k_\sigma - 1$ gives  
    \begin{align} \label{eq:exp_queries_lb1}
        \mathbb{E}[W] &\geq q_1 + \frac{1}{n} \sum_{\sigma \in \mathcal{K}} |C_\sigma| \cdot (k_\sigma - 1) = q_1 + \frac{1}{n} \sum_{\sigma \in \mathcal{K}} |C_\sigma| k_\sigma - \frac{1}{n} \sum_{\sigma \in \mathcal{K}} |C_\sigma| \,.
        \end{align}
 The sets $C_\sigma$ form a partition of $V$, so $\sum_{\sigma \in \mathcal{K}} |C_\sigma| = n$, which substituted in \eqref{eq:exp_queries_lb1} implies  
        \begin{align} \label{eq:expected_W_lb_q_1_and_sum}
        \mathbb{E}[W] &\geq q_1 - 1 + \frac{1}{n} \sum_{\sigma \in \mathcal{K}} |C_\sigma| k_\sigma \,.
    \end{align}

    To obtain a lower  bound, we aim to minimize $\sum_{\sigma \in \mathcal{K}} |C_\sigma| k_\sigma$ subject to $\sum_{\sigma \in \mathcal{K}} k_\sigma \geq cn$. By the Cauchy-Schwarz inequality:
    \[ \left( \sum_{\sigma \in \mathcal{K}} k_\sigma \right)^2 \leq \left( \sum_{\sigma \in \mathcal{K}} \frac{k_\sigma}{|C_\sigma|} \right) \left( \sum_{\sigma \in \mathcal{K}} |C_\sigma| k_\sigma \right) \,. \]
    Since $k_\sigma \leq |C_\sigma|$, we have $\sum_{\sigma \in \mathcal{K}} \frac{k_\sigma}{|C_\sigma|} \leq |\mathcal{K}| \leq q_1 + 1$. Then:
    \begin{align} \label{eq:cn_squared_less_than_component}
    (cn)^2 \leq \left( \sum_{\sigma \in \mathcal{K}} k_\sigma \right)^2 \leq (q_1 + 1) \sum_{\sigma \in \mathcal{K}} |C_\sigma| k_\sigma \implies \sum_{\sigma \in \mathcal{K}} |C_\sigma| k_\sigma \geq \frac{c^2 n^2}{q_1 + 1} \,.
    \end{align}
    Plugging \eqref{eq:cn_squared_less_than_component}  into \eqref{eq:expected_W_lb_q_1_and_sum}:
    \[ \mathbb{E}[W] \geq (q_1 + 1) + \frac{c^2 n}{q_1 + 1} - 2 \,. \]
    The function $g(x) = x + A/x - 2$ is minimized at $x = \sqrt{A}$. Setting $A = c^2 n$ gives
$  \mathbb{E}[W] \geq 2c\sqrt{n} - 2 $ as required.
\end{proof}

\subsection{Lower Bound for $t$ Rounds}

Building on the two-round proof,  we show a lower bound for any number of rounds $t \geq 1$.

\noindent \textbf{Theorem~\ref{thm:lb_t_rounds}.}
\emph{Let $G = (V,E)$ be a connected undirected graph with $n$ vertices. The randomized query complexity of finding a local minimum on $G$ in $t \in \mathbb{N}^*$ rounds is $\Omega(t n^{1/t} - t)$.}
\begin{proof}
We show that for each $c \in (1/n, 1]$, the randomized query complexity of finding a local minimum with success probability at least $c$ is $\Omega(ct n^{1/t} - t)$.

We use Yao's Minimax Principle. Consider the probability distribution $\mathcal{D}$ from Definition~\ref{def:distribution_D}. We  lower bound the expected query complexity of any deterministic $t$-round algorithm $\mathcal{A}$ that succeeds with probability at least $c$ when the input is drawn from $\mathcal{D}$.

By Remark~\ref{rmk:unique_local_min_in_F},
 for each $v \in V$, vertex  $v$ is the unique local  minimum of $f_v$ in  $G$.
 The target local minimum $Z$ is a random variable chosen uniformly at random from $V$. A query at $x$ reveals $f_Z(x)$, which determines whether $x \preceq_T Z$.

\paragraph{Algorithm Execution and Candidate Sets.}

The trajectory of $\mathcal{A}$  when the input is drawn from $\mathcal{D}$ is characterized by the history of interactions. Let $\mathcal{R}_0 = \{V\}$. For $\ell \in [t]$,
let $\mathcal{R}_{\ell}$ be the collection of non-empty candidate sets attainable by the algorithm at the end of round $\ell$.
By Lemma~\ref{lem:candidate_set_characterization}, $\mathcal{R}_{\ell}$ forms a partition of $V$ and each candidate set  $C \in \mathcal{R}_{\ell}$
 is in correspondence to a unique history $H_{\ell}$ reachable by the algorithm at the end of round $\ell$, so $C$ completely determines the round $\ell+1$ query set (denoted $Q_{\ell+1}(C)$).
 
 The oracle's response  in round $\ell+1$ reveals the signature $\sigma = S_{Q_{\ell+1}(C)}(Z)$, which allows $\mathcal{A}$ to restrict the search to the subset of candidates in $C$ that match this signature: $ C_\sigma := \{ v \in C \mid S_{Q_{\ell+1}(C)}(v) = \sigma \}$.
The partition $\mathcal{R}_{\ell+1}$ is the collection of all such non-empty sets $C_\sigma$ obtained from the sets in $\mathcal{R}_{\ell}$.

Moreover, for each $C \in \mathcal{R}_{\ell}$, since the prior distribution $\mathcal{D}$ is uniform over $V$, the posterior distribution of the target $Z$, conditioned on $Z \in C$, remains uniform over $C$.

\paragraph{Branching Factors, Shrinkage, and Expected Cost.}

For $i \in [t]$, let  $C \in \mathcal{R}_{i-1}$ be a candidate set and $Q_i(C)$ the queries submitted in round $i$. Let $q_i(C) = |Q_i(C)|$   and $m_i(C)$ be the number of candidate sets in $\mathcal{R}_i$ that $C$ splits into.  By Lemma~\ref{lem:signatures}, 
\begin{align} \label{eq:q_i_C_at_least_m_i_C_minus_1}
    q_i(C) \ge m_i(C) - 1 \,.
\end{align}
For each $v \in V$, let $\mathcal{C}_{i}(v)$ denote the unique candidate set in $\mathcal{R}_{i}$ containing $v$.
We can write the total number of queries, denoted $\mathcal{W}(Z)$, as:
$ \mathcal{W}(Z) = \sum_{i=1}^t q_i(\mathcal{C}_{i-1}(Z)) \,. $
Taking  expectation over $Z$ drawn uniformly from $V$ and applying inequality \eqref{eq:q_i_C_at_least_m_i_C_minus_1} gives
\begin{align} \label{eq:lb_expected_total_queries_at_least_sum_over_expected_branching_minus_t}
 \E_{Z \sim V}[\mathcal{W}(Z)] =  \sum_{i=1}^t \E_{Z \sim V}\bigl[q_i(\mathcal{C}_{i-1}(Z))\bigr]  &\geq
    \left( \sum_{i=1}^t \E_{Z \sim V}\bigl[ m_i(\mathcal{C}_{i-1}(Z)) \bigr] \right)- t \,.
\end{align}

We define the \emph{shrinkage factor} in round $i$ for input $Z$ as $\rho_i(Z) = \frac{|C_{i-1}(Z)|}{|C_i(Z)|}$.
We now show that the expected branching factor ($\mathbb{E}_{Z \sim V}[ m_i(\mathcal{C}_{i-1}(Z))]$) equals the expected shrinkage factor ($\mathbb{E}_{Z \sim V}[\rho_i(Z)]$).

\paragraph{Showing
$\mathbb{E}_{Z \sim V}[ m_i(\mathcal{C}_{i-1}(Z))] = \mathbb{E}_{Z \sim V}[\rho_i(Z)]$.}
We evaluate the expectations by summing over all $v \in V$.
First, consider the expected branching factor:
\begin{align}
 \mathbb{E}_{Z \sim V} [  m_i(\mathcal{C}_{i-1}(Z))  ] &= \frac{1}{n} \cdot \sum_{v \in V}  m_i(\mathcal{C}_{i-1}(v)) .
\end{align}
We group the summation by the candidate sets in $\mathcal{R}_{i-1}$, obtaining:
\begin{align}  \label{eq:expected_branching_factor}
n \cdot \mathbb{E}_{Z \sim V} [  m_i(\mathcal{C}_{i-1}(Z))  ] &=  \sum_{v \in V}  m_i(\mathcal{C}_{i-1}(v))  = \sum_{C \in \mathcal{R}_{i-1}} \sum_{v \in C} m_i(C)  = \sum_{C \in \mathcal{R}_{i-1}} |C| \cdot m_i(C).
\end{align}
Second, consider the expected shrinkage factor:
$\mathbb{E}_{Z \sim V} [\rho_i(Z)] =  \frac{1}{n} \sum_{v \in V} \rho_i(v) \,.$
For each candidate set $D \in \mathcal{R}_i$, let $\pi(D)$ be the unique candidate set (``parent'') in $\mathcal{R}_{i-1}$ containing $D$. For each $v \in D$, we have $\mathcal{C}_i(v)= D$ and $\mathcal{C}_{i-1}(v)=\pi(D)$, so
\begin{align} \label{eq:expected_shrinkage_factor}
n \cdot \mathbb{E}_{Z \sim V} [\rho_i(Z)]  =
\sum_{v \in V} \rho_i(v) =
\sum_{v \in V} \frac{|\mathcal{C}_{i-1}(v)|}{|\mathcal{C}_i(v)|} = \sum_{D \in \mathcal{R}_i} \sum_{v \in D} \frac{|\pi(D)|}{|D|}
= \sum_{D \in \mathcal{R}_i} |D| \cdot \frac{|\pi(D)|}{|D|} = \sum_{D \in \mathcal{R}_i} |\pi(D)| \,.
\end{align}
For each $C \in \mathcal{R}_{i-1}$ and child $D \in \mathcal{R}_i$ of $C$, we have:
\begin{align} \label{eq:combined_end_terms}
\sum_{D \in \mathcal{R}_i} |\pi(D)| = \sum_{C \in \mathcal{R}_{i-1}} m_i(C) \cdot |C| \,.
\end{align}
Combining \eqref{eq:expected_branching_factor}, \eqref{eq:expected_shrinkage_factor}, and \eqref{eq:combined_end_terms} implies that
\begin{align}  \label{eq:exp_branching_equals_exp_shrinkage}
\mathbb{E}_{Z \sim V} [  m_i(\mathcal{C}_{i-1}(Z))  ]  = \mathbb{E}_{Z \sim V} [\rho_i(Z)]\,.
\end{align}
\paragraph{Simplifying the Lower Bound.}
Using \eqref{eq:exp_branching_equals_exp_shrinkage} in \eqref{eq:lb_expected_total_queries_at_least_sum_over_expected_branching_minus_t} yields  
\begin{align} \label{eq:expectation_sum_over_Z_W_Z_plus_t}
\E_{Z \sim V}[\mathcal{W}(Z)] + t \geq \sum_{i=1}^t \mathbb{E}_{Z \sim V}[\rho_i(Z)] = \mathbb{E}_{Z \sim V} \left[\sum_{i=1}^t \rho_i(Z)\right] \,.
\end{align}
The AM-GM inequality applied pointwise for each $v \in V$ gives  
$ \sum_{i=1}^t \rho_i(v) \ge t \cdot \left(\prod_{i=1}^t \rho_i(v)\right)^{1/t} \,.
$
The product of shrinkage factors telescopes:
\begin{align}
\prod_{i=1}^t \rho_i(v) = \frac{|\mathcal{C}_0(v)|}{|\mathcal{C}_1(v)|} \cdots \frac{|\mathcal{C}_{t-1}(v)|}{|\mathcal{C}_t(v)|} = \frac{|\mathcal{C}_0(v)|}{|\mathcal{C}_t(v)|} = \frac{n}{|\mathcal{C}_t(v)|} \,.
\end{align}
We get
\begin{align} \label{eq:sum_rho_i_over_all_v_lb}
\sum_{i=1}^t \rho_i(v) \geq  t \left(\frac{n}{|\mathcal{C}_t(v)|}\right)^{\frac{1}{t}} \,.
\end{align}
Taking  expectation in \eqref{eq:sum_rho_i_over_all_v_lb} and using \eqref{eq:expectation_sum_over_Z_W_Z_plus_t} gives
\begin{align} \label{eq:lb_t_intermediate}
\E_{Z \sim V}[\mathcal{W}(Z)] + t   \ge \mathbb{E}_{Z \sim V} \left[ t \left(\frac{n}{|\mathcal{C}_t(Z)|}\right)^{\frac{1}{t}} \right] = t \cdot n^{1/t} \cdot \mathbb{E}_{Z \sim V} \left[|\mathcal{C}_t(Z)|^{-\frac{1}{t}}\right].
\end{align}

\paragraph{Success Condition and Optimization.}
Denote the set of vertices for which $\mathcal{A}$ succeeds by  
\begin{align}
\Psi = \left\{ v \in V \mid \mathcal{A} \mbox{ outputs } v \mbox{ on input } f_v \right\}\,.
\end{align}
If a final candidate set $C \in \mathcal{R}_t$ contains at least two distinct vertices $u$ and $v$, then $\mathcal{A}$ observed identical histories for $f_u$ and $f_v$, so it produced the same output on both even though  their local minima are distinct ($u$ for  $f_u$ and $v$ for $f_v$). Thus  $\mathcal{A}$ cannot succeed on both, so
\begin{align}  \label{eq:C_cap_Psi_at_most_one}
|C \cap \Psi| \le 1 \; \forall C \in \mathcal{R}_t \,.
\end{align}

The total number of final candidate sets $L := |\mathcal{R}_t|$  satisfies
\begin{align}
L = \sum_{C \in \mathcal{R}_t} 1 \ge \sum_{C \in \mathcal{R}_t} |C \cap \Psi| = | \Psi | \,. 
\end{align}
The success probability of $\mathcal{A}$ on distribution  $\mathcal{D}$ is $P_{succ} := |\Psi|/n \geq c$. Thus $L \geq |\Psi| \geq \lceil cn \rceil =: L_{\min}$.

We seek a lower bound on $M := \mathbb{E}_{Z\sim V}\left[|\mathcal{C}_t(Z)|^{-\frac{1}{t}}\right] = \frac{1}{n} \sum_{v \in V} |\mathcal{C}_t(v)|^{-\frac{1}{t}}$.  Let $\gamma = 1 - 1/t$. Since $t \ge 1$ and $0 \le \gamma < 1$,  
\begin{align} \label{eq:identity_nM_sum_of_C_to_power_gamma}
n M = \sum_{v \in V} |\mathcal{C}_t(v)|^{-\frac{1}{t}} = \sum_{C \in \mathcal{R}_t} \sum_{v \in C} |C|^{-\frac{1}{t}} = \sum_{C \in \mathcal{R}_t} |C|^\gamma \,.
\end{align}
To obtain a lower bound on the query complexity we solve the following optimization problem:
\begin{equation} \label{eq:optimization_problem}
\begin{aligned}
& \underset{\mathcal{R}_t}{\text{minimize}}
& & \sum_{C \in \mathcal{R}_t} |C|^\gamma \\
& \text{subject to}
& & \sum_{C \in \mathcal{R}_t} |C| = n, \\
& & & |C| \ge 1, \quad \forall C \in \mathcal{R}_t, \\
& & & |\mathcal{R}_t| \geq L_{\min} = \lceil cn \rceil.
\end{aligned}
\end{equation}

\emph{Case 1: $t=1$.} The objective function simplifies to counting the sets:
$ \sum_{C \in \mathcal{R}_1} |C|^\gamma = \sum_{C \in \mathcal{R}_1} 1 = L \,. $
Problem~\eqref{eq:optimization_problem} reduces to minimizing $L$ subject to $L \ge \lceil cn \rceil$. The minimum is attained at $ L = \lceil cn \rceil$, which substituted back into the query complexity bound \eqref{eq:lb_t_intermediate} gives
\begin{align}
\E_{Z \sim V}[\mathcal{W}(Z)] + 1 &\ge 1 \cdot n^{1/1} \cdot \E_{Z \sim V}\left[ |\mathcal{C}_1(Z)|^{-1/1} \right]  = n \cdot \frac{1}{n} \sum_{C \in \mathcal{R}_1} |C|^0  \ge \lceil cn \rceil \,.
\end{align}
Therefore, $\E_{Z \sim V}[\mathcal{W}(Z)] \ge \lceil cn \rceil - 1$, which provides the lower bound $\Omega(cn)$.

\emph{Case 2: $t \geq 2$.}  In this case $0 < \gamma < 1$.  
To lower bound the objective of problem \eqref{eq:optimization_problem}, we  consider a continuous version of the optimization problem and proceed in two steps.

\begin{itemize}
\item \emph{Step 1: Optimization for a fixed number of final candidate sets $L$.}
Let $\vec{x} = (x_1, \dots, x_L)$ represent the sizes of the $L$ sets in $\mathcal{R}_t$. Consider the  continuous optimization problem:
\begin{equation} \label{eq:vector_optimization}
\begin{aligned}
& \underset{\vec{x} \in \mathcal{B}}{\text{minimize}}
& & \Phi(\vec{x}) = \sum_{i=1}^L x_i^\gamma, \\
& \text{where}
& & \mathcal{B} = \left\{ \vec{x} \in \mathbb{R}^L \;\middle|\;  x_i \ge 1 \; \forall i \in [L] \mbox{ and } \sum_{i=1}^L x_i = n \right\}.
\end{aligned}
\end{equation}
We determine the convexity of $\Phi$ by computing its Hessian  $H(\vec{x})$. For $i,j \in [L]$, the second partial derivatives are:
\[ \frac{\partial^2 \Phi}{\partial x_i \partial x_j} = \begin{cases} \gamma(\gamma-1)x_i^{\gamma-2} & \text{if } i=j \\ 0 & \text{if } i \neq j \end{cases} \,. \]
Since $0 < \gamma < 1$ and $x_i \geq 1$ $\forall i \in [L]$, we have $\gamma(\gamma-1) < 0$ and $x_i^{\gamma-2} > 0$.
 Thus the Hessian  is a diagonal matrix with strictly negative entries, meaning it is negative definite everywhere on $\mathcal{B}$. Thus $\Phi$ is strictly concave on $\mathcal{B}$.
 Since $\mathcal{B}$ is a bounded convex polytope, there is a  global minimum at a vertex of the polytope.
The vertices of $\mathcal{B}$ are tuples where $L-1$ variables take the minimum value $1$, and the remaining variable takes the value $n - (L-1)$. Due to symmetry, all such vertices yield the same value for the objective.

Let $G : \mathbb{R} \to \mathbb{R}$ be the function
$     G(x) = (x-1) \cdot 1^\gamma + (n - x + 1)^\gamma = x - 1 + (n - x + 1)^\gamma \,.
$
Then for a fixed $L$ the minimum value of the objective in problem \eqref{eq:vector_optimization} is $G(L)$.

\item \emph{Step 2: Optimization over the number of sets $L$.}
We now minimize $G(L)$ with respect to $L$ subject to the constraint $L \ge \lceil cn \rceil$. Treating $L$ as a continuous variable, we analyze the first derivative:
\begin{align}
G'(L) =  1 - \gamma(n - L + 1)^{\gamma-1} = 1 - \frac{1 - {1}/{t}}{(n - L + 1)^{1/t}}  \,.
\end{align}
Since $L \le n$, we have  $(n - L + 1)^{1/t} \ge 1$.  Thus $G'(L) > 0$ for all $L \in [\lceil cn \rceil, n]$, so  $G$ is strictly increasing on this interval. Its  minimum is attained at  $L_{\min} = \lceil cn \rceil$.

This gives the lower bound on the objective:
\begin{align} \label{eq:lb_nM_case2}
 \sum_{C \in \mathcal{R}_t} |C|^\gamma \ge G(L_{\min}) = L_{\min} - 1 + (n - L_{\min} + 1)^{1-1/t}  \,.
\end{align}
\end{itemize}

Using \eqref{eq:lb_nM_case2} and \eqref{eq:identity_nM_sum_of_C_to_power_gamma} in \eqref{eq:lb_t_intermediate}, we can lower bound the expected query complexity as follows:
\begin{align} \label{eq:lb_expected_queries_case_2_almost_done}
\E_{Z \sim V}[\mathcal{W}(Z)] + t &\geq t \cdot n^{1/t} \cdot \mathbb{E}_{Z \sim V} \left[|\mathcal{C}_t(Z)|^{-\frac{1}{t}}\right] = t \cdot n^{1/t} \cdot \frac{1}{n} \sum_{v \in V} |\mathcal{C}_t(v)|^{-\frac{1}{t}}
=  t \cdot n^{1/t} \cdot  \frac{1}{n} \sum_{C \in \mathcal{R}_t} |C|^\gamma \notag \\
&
\geq  t \cdot n^{1/t} \cdot \frac{G(L_{\min})}{n} = t \cdot n^{{1}/{t}-1} (L_{\min} - 1) + t \cdot n^{{1}/{t}-1} (n - L_{\min} + 1)^{1-1/t} \,.
\end{align}

Since $cn \le L_{\min} = \lceil cn \rceil < cn + 1$, we have
\begin{align} \label{eq:inequalities_first_and_second_part_of_last_term}
    t n^{\frac{1}{t}-1} (L_{\min} - 1) \ge  c t n^{1/t} - t n^{\frac{1}{t}-1} \; \;  \mbox{ and } \; \; t n^{\frac{1}{t}-1} (n - L_{\min} + 1)^{1-1/t} \ge t (1-c)^{1-1/t} \,.
\end{align}
Using \eqref{eq:inequalities_first_and_second_part_of_last_term} in \eqref{eq:lb_expected_queries_case_2_almost_done}, we obtain:
$ \E_{Z \sim V}[\mathcal{W}(Z)] \ge c t n^{1/t} + t(1-c)^{1-1/t} - t - t n^{\frac{1}{t}-1} \,.$
Since $t n^{\frac{1}{t}-1} \to_{n \to \infty} 0$, the randomized query complexity is $\Omega(t \cdot c \cdot n^{1/t} - t)$.
\end{proof}

\section*{Methodology}

This work was developed through an interactive collaboration with Google's Gemini-based models, in particular, Gemini Deep Think and its advanced Google-internal variants. The results presented here are the product of a ``scaffolded'' reasoning process, where the authors provided the high-level conceptual direction and  lemma statements, while the model synthesized initial proofs which we then rigorously verified and refined. 

For the deterministic upper bound (Theorem~\ref{thm:k_rounds}), the authors asked the model to exploit the separation number of graphs; in response, the model successfully found a separation-based strategy for two rounds and generalized it to  $t \geq 2$ rounds, identifying the  recursive decomposition and  the   folklore ``Shattering Lemma'' (Lemma~\ref{lem:shattering}). 

The development of the randomized lower bounds (Theorem~\ref{thm:lb_t_rounds}) involved an iterative feedback loop. We asked the model to  construct a lower bound for trees in two rounds while giving it several papers from prior work as examples (\cite{branzei2022query,santha2004quantum}). This led the model  to propose the specific distance-based  function family $\mathcal{F}$. We then asked it to generalize this construction to handle local search on any graph in two rounds;
in response, it proposed using the lower bound for trees by fixing first a spanning tree of the original graph.
This proof was then generalized to $t$ rounds.

Verification and correction were a critical aspect of this process; e.g. the authors identified a circular argument in the model's initial proof for the bijection between histories and candidate sets (part b of Lemma~\ref{lem:candidate_set_characterization}) and provided the model with a hint for resolving it. Finally, the model served as an adversarial check on some of our hypotheses, synthesizing the parallel steepest descent algorithm (Proposition~\ref{thm:t_rounds_const_deg_generalized}) to demonstrate that linear lower bounds do not hold for local search in two rounds on constant-degree expanders. Our paper adds to a growing body of literature on TCS and math written with the help of AI models~\cite{georgiev2025mathematicalexplorationdiscoveryscale,nagda2025reinforcedgenerationcombinatorialstructures,bubeck2025earlyscienceaccelerationexperiments,lu2025solvinginequalityproofslarge,chervov2025cayleypy,sellke2025learningcurvemonotonicitymaximumlikelihood,sothanaphan2026resolutionerdhosproblem728}.

\bibliographystyle{alpha}

\bibliography{local_search_bib}

\appendix

\section{Appendix: Algorithms} \label{app:algorithms}

In this section we include the proofs of several lemmas from Section~\ref{sec:algorithms}.

\subsection{Deterministic Algorithm}

The next lemma is a folklore result; we include its statement and proof for completeness.

\textbf{Lemma~\ref{lem:shattering} (restated).} \emph{Let $G=(V,E)$ be a graph with $n$ vertices and separation number $s$. For any parameter $K \in [1, n]$, there exists a subset of vertices $S \subseteq V$  such that every connected component of the induced subgraph $G[V \setminus S]$ has size at most $K$, and  $|S| < 3sn/K$.}
\begin{proof}
We define a recursive procedure to construct the accumulated separator set $S$.

\paragraph{The Recursive Decomposition Algorithm.}
We define a function \texttt{FindSeparator(H, K)} that takes an induced subgraph $H$ of $G$ and the parameter $K$, and returns a separator set $S_H \subseteq V(H)$:

\begin{enumerate}
    \item If $|V(H)| \le K$, return $S_H = \emptyset$.
    Else  $|V(H)| > K$. Since $H$ is a subgraph of $G$, its separation number is at most $s$. There exists an $(s, 2/3)$-separator $R_0$ in $H$ ($|R_0| \le s$). Let $V(H) = A \cup B \cup R_0$ be the corresponding partition, where $|A|, |B| \le (2/3)|V(H)|$ and $E(A,B) = \emptyset$.
    \item \emph{Recursive Calls:} Let
    $S_A = \texttt{FindSeparator}(G[A], K)$ and
    $S_B = \texttt{FindSeparator}(G[B], K)$.
    \item \emph{Combine:} Return $S_H = R_0 \cup S_A \cup S_B$.
\end{enumerate}

We execute this algorithm starting with $H=G$. Let $S$ be the returned set.

\paragraph{Correctness (Component Size).}
We prove by induction on the recursion depth that the call $\texttt{FindSeparator(H, K)}$ returns a set $S_H$ such that all connected components of $G[V(H) \setminus S_H]$ have size at most $K$.

\begin{itemize}
    \item \emph{Base Case:} If $|V(H)| \le K$, the algorithm returns $S_H = \emptyset$. The remaining subgraph is $G[V(H)]$. Since the entire subgraph has at most $K$ vertices, any connected component within it must also have size at most $K$.

    \item \emph{Inductive Step:} If $|V(H)| > K$, the algorithm finds separator $R_0$ and balanced components $A$, $B$ such that  $E(A,B) = \emptyset$ and returns $S_H = R_0 \cup S_A \cup S_B$. We have \[
    V(H) \setminus S_H = (A \cup B \cup R_0) \setminus (R_0 \cup S_A \cup S_B) = (A \setminus S_A) \cup (B \setminus S_B).\]
    Since $E(A,B) = \emptyset$, there are no edges between $A \setminus S_A$ and $B \setminus S_B$. Therefore, any connected component of $G[V(H) \setminus S_H]$ must be a connected component of \emph{either} $G[A \setminus S_A]$ \emph{or} $G[B \setminus S_B]$.
    By the inductive hypothesis on the recursive calls, the components of $G[A \setminus S_A]$ and $G[B \setminus S_B]$ are all of size at most $K$. This completes the inductive step.
\end{itemize}

The  correctness follows from the top-level call with $H=G$.

\paragraph{Analysis of the Size of S (Charging Argument).}
We analyze the total size of the ``accumulated'' separator $S$ using an amortized analysis (a charging argument). The sum of the sizes of all the individual separators ($R_0$) introduced at every step of the recursion is $|S|$.

We distribute the cost of the separator among the vertices. When processing a subgraph $H$ with  $|V(H)|>K$, we introduce a separator $R_0$ of size $|R_0| \le s$. We charge this cost equally to the vertices in $V(H)$. The charge per vertex $v \in V(H)$ at this step is defined as:
\begin{align}
\text{Charge}(v, H) := \frac{|R_0|}{|V(H)|} \le \frac{s}{|V(H)|}.
\end{align}

The total size of the accumulated separator $|S|$ is equal to the total charge accumulated across all steps. Let $C(v)$ be the total charge accumulated by vertex $v$ throughout the entire process. Then $|S| = \sum_{v \in V} C(v)$.

We now bound $C(v)$ for an arbitrary vertex $v$. Consider the sequence of subgraphs $H_0, H_1, \ldots$ in the recursion tree that contain $v$, where $H_0 = G$.

For each $i \ge 0$, the algorithm's action on $H_i$ depends on its size:
\begin{itemize}
    \item If $|V(H_i)| \le K$ (base case), no separator is created, $v$ accumulates no charge, and this sequence for $v$ terminates.
    \item If $|V(H_i)| > K$, the algorithm processes $H_i$. It finds a separator (denoted $R_i$), such that  $|R_i| \le s$, the remaining parts $A_i,B_i$ satisfy  $|A_i|, |B_i| \leq 2/3 |V(H_i)|$, $E(A_i, B_i) = \emptyset$, and $V(H_i) = A_i \cup B_i \cup R_i$. At this step, $v$ accumulates a charge $\text{Charge}(v, H_i) \le s / |V(H_i)|$.
    \begin{itemize}
        \item If $v \in R_i$, then $v$ is removed, and the sequence for $v$ terminates.
        \item If $v \notin R_i$, then $v$ is in either $A_i$ or $B_i$. We define $H_{i+1}$ as the subgraph $G[A_i]$ or $G[B_i]$ that contains $v$, and the process continues.
    \end{itemize}
\end{itemize}

Let $t$ be the index of the last subgraph in this sequence for which $|V(H_i)| > K$. The vertex $v$ only accumulates charges from these subgraphs $H_0, \ldots, H_t$. The total charge is bounded by:
\begin{align} \label{eq:C_V_upper_bound}
C(v) \le \sum_{i=0}^{t} \frac{s}{|V(H_i)|}.
\end{align}

Crucially, by the definition of the $(s, 2/3)$-separator, the size of the subgraphs in this sequence decreases geometrically with
$|V(H_{i+1})| \le ({2}/{3}) |V(H_i)|$, and so
\begin{align}
|V(H_i)| \ge \left(\frac{3}{2}\right)^{t-i} |V(H_{t})|.
\end{align}
Since $H_{t}$ was processed, we know $|V(H_{t})| > K$. Thus, $|V(H_i)| > (3/2)^{t-i} K$.
We substitute this lower bound back into inequality \eqref{eq:C_V_upper_bound}:
\begin{align} \label{eq:precise_ub_on_charge_of_v}
C(v) & \leq \sum_{i=0}^{t} \frac{s}{|V(H_i)|} < \sum_{i=0}^{t} \frac{s}{(3/2)^{t-i} K} = \frac{s}{K} \sum_{i=0}^{t} \left(\frac{2}{3}\right)^{t-i} = \frac{s}{K} \sum_{j=0}^{t} \left(\frac{2}{3}\right)^{j} \leq \frac{3s}{K}\,.
\end{align}

Using \eqref{eq:precise_ub_on_charge_of_v} we can  bound the total size of the accumulated separator $S$ by
$|S| = \sum_{v \in V} C(v) < n \cdot \frac{3s}{K}$, which completes the proof.
\end{proof}

\noindent \textbf{Lemma~\ref{claim:non_exploration} (restated).} \emph{In the setting of Theorem~\ref{thm:k_rounds}, let $i \in \{2, \ldots, t\}$ and $C \in \mathcal{C}_{i-1}$. If the running minimum $v_{i-1}$ has no neighbors in $C$, then no vertex in $C$ is ever queried by the algorithm.}
\begin{proof}
    We prove by induction on the round $r \in \{i, \ldots, t\}$ that $Q_r \cap C = \emptyset$ and that $v_r$ has no neighbors in $C$.

    \emph{Base Case ($r=i$):}
    By assumption $v_{i-1}$ has no neighbors in $C$.
    In round $i$, the algorithm finds the set of components $\mathcal{A}_{i-1} \subseteq \mathcal{C}_{i-1}$  adjacent to $v_{i-1}$.
    Since $v_{i-1}$ is not adjacent to $C$, we get $C \notin \mathcal{A}_{i-1}$.
    Thus, any component $D \in \mathcal{A}_{i-1}$ selected for querying must be distinct from $C$. By the decomposition property, distinct components at level $i-1$ are disjoint and disconnected (i.e., $E(D, C) = \emptyset$).
    Because the query set $Q_i$ is a subset of the union of these selected components ($Q_i \subseteq \bigcup_{D \in \mathcal{A}_{i-1}} D$), we get that $Q_i \cap C = \emptyset$ and no vertex in $Q_i$ is adjacent to $C$.
    Finally, since neither $v_{i-1}$ nor any vertex in $Q_i$ is adjacent to $C$, the updated minimum $v_i \in \{v_{i-1}\} \cup Q_i$ has no neighbors in $C$.

    \emph{Inductive Step ($r > i$):}
    Assume the hypothesis holds for round $r-1$, so $v_{r-1}$ has no neighbors in $C$.
    In round $r$, the algorithm queries a component $D \in \mathcal{C}_{r-1}$ only if it is adjacent to $v_{r-1}$. By the hierarchical decomposition, $D$ is a subgraph of a unique ancestor component $A \in \mathcal{C}_{i-1}$.
    \begin{itemize}
        \item If $A = C$, then $D \subseteq C$. Adjacency to $D$ would imply $v_{r-1}$ is adjacent to $C$, contradicting the inductive hypothesis. Thus, no sub-component of $C$ is selected ($Q_r \cap C = \emptyset$).
        \item If $A \neq C$, then $D \subseteq A$. Since distinct components at level $i-1$ are disconnected ($E(A, C) = \emptyset$), no vertex in $D$ is adjacent to $C$.
    \end{itemize}
    Thus, $Q_r$ contains no vertices in $C$ and no vertices adjacent to $C$. Since $v_{r-1}$ is not adjacent to $C$ (by hypothesis), the updated minimum $v_r \in \{v_{r-1}\} \cup Q_r$ has no neighbors in $C$.
\end{proof}

\begin{lemma} \label{lem:ub_t_minus_one}
Let $x \in \mathbb{R}$ and $t \in \mathbb{N}$ such that $x > 1$ and $t \geq 2$. Then
$
\frac{1 - x^{2/t - 1}}{x^{1/t} - 1} \leq t - 1\,.
$
\end{lemma}
\begin{proof}
 Let $y := x^{1/t}$. Then $y  > 1$ since $x > 1$.
Substituting $x = y^t$, the  inequality becomes
$
\frac{1 - y^{2-t}}{y - 1} \leq t-1\,.
$
This is immediate for $t = 2$, so  assume $t > 2$. Let $k = t - 2$.

Expanding the fraction as a geometric series yields
$
\frac{1 - y^{-k}}{y - 1} = \sum_{j=1}^{k} y^{-j}\,.
$
Then $y^{-j} < 1$ since $y > 1$, so the sum is strictly bounded by the number of terms, $k$. Thus
$
\frac{1 - y^{2-t}}{y - 1} < k = t - 2 < t - 1
$.
\end{proof}

\subsection{Randomized Algorithm}

\noindent \textbf{Lemma~\ref{lem:expected_rank_replacement} (restated).} \emph{Let $Q$ be a multiset of $q$ vertices sampled uniformly at random with replacement from $V$. Let $v_{\min}$ be the minimum vertex in $Q$ with respect to $\prec$. Then }
\begin{align} \E[\mathcal{L}(v_{\min})] < \frac{n}{q+1} \,.
\end{align}
\begin{proof}
Let $R = \Rank(v_{\min})$. Since the steepest descent path strictly decreases in rank at every step, its length is bounded by the starting rank: $\mathcal{L}(v_{\min}) \le R - 1$. Thus it suffices to show  $\E[R] < \frac{n}{q+1} + 1$.

We derive the bound starting from the definition of the expected value for the integer-valued random variable $R \in \{1, \dots, n\}$:
    $\E[R] = \sum_{k=1}^n k \cdot \Pr(R = k) =  \sum_{j=1}^n \Pr(R \ge j).$

The event $R \ge j$ occurs if and only if all $q$ sampled vertices have a rank of at least $j$. Since there are $n - (j - 1)$ such vertices, the probability is:
$
    \Pr(R \ge j) = \left( 1 - \frac{j-1}{n} \right)^q.
$
Substituting this  into the expectation sum with the change of variable $i = j-1$ gives:
$
    \E[R] = \sum_{i=0}^{n-1} \left( 1 - \frac{i}{n} \right)^q.
$
Let $g: \mathbb{R} \to \mathbb{R}$ be $g(x) = (1 - x/n)^q$. Since $g$ is  strictly decreasing on $[0, n]$, for each  $i \in \mathbb{N}^*$,  
$
    g(i) < \int_{i-1}^{i} g(x) \, dx.
$
Then we can bound the expectation by
\begin{align}
    \E[R] = g(0) + \sum_{i=1}^{n-1} g(i) < 1 + \sum_{i=1}^{n-1} \int_{i-1}^{i} g(x) \, dx 
    < 1 + \int_{0}^{n} \left( 1 - \frac{x}{n} \right)^q \, dx.
\end{align}

Setting $u = 1 - x/n$, we have:
$
    \int_{0}^{n} \left( 1 - \frac{x}{n} \right)^q \, dx 
    = \frac{n}{q+1}.
$
Therefore $\E[R] < 1 + \frac{n}{q+1}$, as required.
\end{proof}

\noindent \textbf{Lemma~\ref{lem:ball_size_rigorous} (restated).} \emph{Let $G = (V,E)$ be a graph with maximum degree $\Delta \ge 2$. Let $v \in V$ and  $\rho \in \mathbb{N}^*$.  
      If $\Delta=2$, then $|B(v; \rho)| \le 2\rho+1$.
        Else if $\Delta \ge 3$, then $|B(v; \rho)| < \frac{\Delta}{\Delta-2}(\Delta-1)^\rho$.}
\begin{proof}
    Let $n_k$ denote the number of vertices in distance exactly $k$ from $v$. We have $n_0 = 1$ (containing just $v$) and $n_1 \le \Delta$ since $v$ has at most $\Delta$ neighbors.
   
    For $k \ge 2$, let  $u$ be a vertex in  distance $k-1$ from $v$. Then  $u$ has an edge  to a vertex $w$ in distance $k-2$ from $v$. Thus $u$ has at most $\Delta - 1$ remaining edges connecting to vertices at distance $k$ from $v$. This yields the recurrence:
    $
        n_k \le n_{k-1}(\Delta - 1)$, so $n_k \le \Delta(\Delta-1)^{k-1}$.
    The size of the ball is the sum of the sizes of these layers:
    \begin{align}  \label{eq:ub_ball_rho}
        |B(v; \rho)| = n_0 + \sum_{k=1}^\rho n_k \le 1 + \sum_{k=1}^\rho \Delta(\Delta-1)^{k-1}.
    \end{align}
    Inequality \eqref{eq:ub_ball_rho} implies $|B(v; \rho)| \le 1 + 2\rho$ when $\Delta=2$ and $|B(v; \rho)| \leq \frac{\Delta(\Delta-1)^\rho}{\Delta-2}$ when $\Delta \geq 3$.
\end{proof}

\noindent \textbf{Proposition~\ref{thm:t_rounds_const_deg_generalized} (restated).}
\emph{Let $G=(V,E)$ be a graph with $n$ vertices and maximum degree $\Delta$. The randomized query complexity of finding a local minimum in $t \ge 2$ rounds is $O(\sqrt{n} + t)$ when $\Delta \leq 2$ and $O\bigl(\frac{n}{t \cdot \log_\Delta n} + t \Delta^2 \sqrt{n}\bigr)$ when $\Delta \geq 3$.}
\begin{proof}
    We analyze Algorithm~\ref{alg:wssd} and later set $q_1$ and $r$ to obtain the required bounds. Let $T = t-1$ be the number of parallel search rounds. Let $W$ be the random variable for the total query cost, consisting of the sampling cost $q_1$ and the queries performed in the $T$ search steps. We have:
    \begin{equation} \label{eq:total_cost_upper_bound}
        W \le q_1 + \sum_{i=1}^{T} |B(v^{(i-1)}; r+1)| \le q_1 + T \cdot \max_{v \in V} |B(v; r+1)| \,.
    \end{equation}
    For each  $j \in \{0, 1, \ldots, T\}$, we have $dist(v^{(0)}, v^{(j)}) = j \cdot r$.
    The algorithm succeeds in finding a local minimum when the steepest descent path starting from the minimal sampled vertex $v^{(0)}$ has length $\mathcal{L}(v^{(0)}) \leq T \cdot r$. Applying Markov's inequality the non-negative random variable $\mathcal{L}(v^{(0)})$ and using Lemma~\ref{lem:expected_rank_replacement}, we can bound the probability of the failure event:
    \begin{align}
        \Pr(\text{Failure}) \leq \Pr\bigl(\mathcal{L}(v^{(0)}) > T r\bigr) \le  \Pr\bigl(\mathcal{L}(v^{(0)}) \geq T r\bigr) \leq \frac{\E[\mathcal{L}(v^{(0)})]}{T r} < \frac{n}{Tr(q_1+1)} \,.
    \end{align}
    To ensure  failure probability less than $1/10$, we require that  
        $(q_1 + 1) T r \ge 10n$ ($\dag$).

    \paragraph{Case 1: $\Delta \leq 2$.} The statement is immediate if  $\Delta=1$, so let $\Delta =2$.
    Set
    $
        q_1 = \left\lceil \sqrt{20n} \, \right\rceil$ and $r = \bigl\lceil \frac{\sqrt{5n}}{T} \bigr\rceil \,.
    $
    Then
    $
        (q_1 + 1) T r > \sqrt{20n} \cdot T \cdot \frac{\sqrt{5n}}{T} = 10n
    $, so ($\dag$) holds.
   
     Lemma~\ref{lem:ball_size_rigorous} yields  $\max_{v \in V} |B(v; r+1)| \leq 2r+3$, which combined with \eqref{eq:total_cost_upper_bound} bounds the query cost:
    \begin{align}
        W &\le q_1 + T(2r + 3)
        < (\sqrt{20n} + 1) + 2T \Bigl( \frac{\sqrt{5n}}{T} + 1 \Bigr) + 3T
        \in  O(\sqrt{n} + t) \,.
    \end{align}
    \paragraph{Case 2: $\Delta \ge 3$.}
     Set  
    $
        r = \left\lceil \frac{1}{2} \log_{(\Delta-1)} n \right\rceil $ and $ q_1 = \left\lceil \frac{10n}{T r} \right\rceil \,.$
    Then $(q_1+1)Tr > q_1 Tr \ge 10n$, so  ($\dag$) holds.
    By Lemma~\ref{lem:ball_size_rigorous},
    \begin{align}  \label{eq:ub_max_ball_size}
    \max_{v \in V} |B(v; r+1)| <  \frac{\Delta}{\Delta-2} (\Delta - 1)^{r+1} \leq 3(\Delta-1)^{r+1} \,.
    \end{align}
    We now bound the two components of the cost.
    \begin{enumerate}[(a)]
\item \emph{Search cost.} By choice of $r$, we get   
       $ (\Delta-1)^{r+1} < (\Delta-1)^{\frac{1}{2} \log_{(\Delta-1)} n + 2}
        < \Delta^2 \sqrt{n} \,.$
     With \eqref{eq:ub_max_ball_size},  this bounds the  search cost to: 
        $T \max_{v \in V} |B(v; r+1)| \leq  3T(\Delta-1)^{r+1} < 3 T \Delta^2 \sqrt{n} \in O(t \Delta^2 \sqrt{n}) \,.$
   
\item \emph{Sampling cost.}
    The sampling cost is:
    $
        q_1 \le \frac{10n}{T (\frac{1}{2} \log_{(\Delta-1)} n)} + 1 \in O\left( \frac{n}{t \log_\Delta n} \right) \,.
    $
\end{enumerate}
    Summing both components  yields the bound $O\left( \frac{n}{t \log_\Delta n} + t \Delta^2 \sqrt{n} \right)$.
\end{proof}

\end{document}